\documentclass[lettersize,journal]{IEEEtran}
\usepackage{etoolbox}
\usepackage{algorithmic}
\usepackage{algorithm}
\usepackage{amsmath,amsfonts,flushend}
\usepackage{algorithmic}

\usepackage{array}

\usepackage{textcomp}
\usepackage{stfloats}
\usepackage{url}
\usepackage{verbatim}
\usepackage{graphicx}
\usepackage{cite}
\usepackage{amsmath,amssymb}

\usepackage{mathtools, cuted}
\usepackage{graphicx,graphics,color,psfrag}
\usepackage{cite,balance}
\usepackage{caption}
\allowdisplaybreaks
\usepackage{subfigure}
\usepackage{accents}
\usepackage{amsthm}
\usepackage{bm}
\usepackage{algorithmic}
\usepackage[english]{babel}
\usepackage{multirow}
\usepackage{enumerate}
\usepackage[T1]{fontenc}
\usepackage{aecompl}
\usepackage{cases}
\usepackage{dsfont}
\usepackage{color,soul}

\usepackage{lipsum}
\usepackage{fancyhdr}
\usepackage{hhline}
\usepackage{hyperref}
\hypersetup{
    colorlinks=true, 
    urlcolor=blue,   
}
\newtheorem{theorem}{\textbf{Theorem}}

\newtheorem{lemma}{\textbf{Lemma}}

\newtheorem{proposition}{\textbf{Proposition}}
\newtheorem{remark}{\textbf{Remark}}

\newtheorem{assumption}{Assumption}

\begin{document}
\bibliographystyle{IEEEtran}

\title{ICDM: Interference Cancellation Diffusion Models for Wireless Semantic Communications}

\author{Tong Wu, Zhiyong Chen, Dazhi He, Feng Yang,  Meixia Tao, \emph{Fellow, IEEE}, Xiaodong Xu, \\Wenjun Zhang, \emph{Fellow, IEEE}, and Ping Zhang, \emph{Fellow, IEEE}

\thanks{T. Wu, Z. Chen, D. He, F. Yang,  M. Tao, and W. Zhang are with the Cooperative Medianet Innovation Center (CMIC), Shanghai Jiao Tong University, Shanghai 200240, China. F. Yang, M. Tao and W. Zhang are also with the Department of Electronic Engineering, Shanghai Jiao Tong University, Shanghai 200240, China. (e-mail: \{wu\_tong, zhiyongchen, hedazhi, yangfeng, mxtao, zhangwenjun\}@sjtu.edu.cn).}
\thanks{X. Xu and P. Zhang are with the State Key Laboratory of Networking and Switching Technology, Beijing University of Posts and Telecommunications, Beijing 100876, China, and also with the Department of Broadband Communication, PengCheng Laboratory, Shenzhen 518055, China (e-mail: xuxd@pcl.ac.cn; pzhang@bupt.edu.cn).}}



\maketitle

\begin{abstract}
Diffusion models (DMs) have recently achieved significant success in wireless communications systems due to their denoising capabilities. The broadcast nature of wireless signals makes them susceptible not only to Gaussian noise, but also to unaware interference. This raises the question of whether DMs can effectively mitigate interference in wireless semantic communication systems. In this paper, we model the interference cancellation problem as a maximum a posteriori (MAP) problem over the joint posterior probability of the signal and interference, and theoretically prove that the solution provides excellent estimates for the signal and interference. To solve this problem, we develop an interference cancellation diffusion model (ICDM), which decomposes the joint posterior into independent prior probabilities of the signal and interference, along with the channel transition probablity. The log-gradients of these distributions at each time step are learned separately by DMs and accurately estimated through deriving. ICDM further integrates these gradients with advanced numerical iteration method, achieving accurate and rapid interference cancellation. Extensive experiments demonstrate that ICDM significantly reduces the mean square error (MSE) and enhances perceptual quality compared to schemes without ICDM. For example, on the CelebA dataset under the Rayleigh fading channel with a signal-to-noise ratio (SNR) of $20$ dB and signal to interference plus noise ratio (SINR) of $0$ dB, ICDM reduces the MSE by $4.54$ dB and improves the learned perceptual image patch similarity (LPIPS) by $2.47$ dB. 
\end{abstract}


\section{Introduction}
Diffusion models (DMs)\cite{DDPM2015,ddpm,smld,projection}, which utilize a score function to model the gradient of the conditional data distribution, have recently achieved remarkable success in the field of artificial intelligence generated content (AIGC). These models are capable of generating controllable and high-quality content in various domains, including long-form text generation, controllable image generation, and consistent video generation. They have also become a fundamental technology for large language models (LLMs) such as GPT-4o. The inherent controllability of the content generated by diffusion models has significantly driven their application across diverse fields\cite{Liang,LAMSC,LLM6G}.

Recently, some research has explored the potential of DMs in wireless communications \cite{Liang,Sengupta}. For example. SING \cite{SING} adopts DMs to generate contents in the approximate null-space of reconstructed images, thereby enhancing image quality. SG2SC \cite{shiguangming} utilizes DMs to generate complete images based on diverse semantic, achieving high-quality image reconstruction with semantic consistency. In \cite{Sergio}, segmentation maps are transmitted to the receiver, which then reconstructs the original image using DMs conditioned on the map, significantly reducing transmission cost while preserving semantic similarity. \cite{DiffCom} treats the received signal as a natural conditioning input and directly applies DMs to generate images based  on the received signal. VLM-CSC in \cite{VLM-CSC} transmits the text description consistent with the image to the receiver, and then the receiver generates the image conditioned on the text with DMs.


Importantly, \cite{CDDM} highlights a fundamental similarity between DMs and wireless communication systems: both aim to recover original data from input signals corrupted by Gaussian noise. Building on this insight, \cite{CDDM} proposes CDDM, which employs DMs as a novel physical module to gradually denoising received signals, thereby enabling higher-quality decoding. This approach has further been extended to address challenges such as denoising in MIMO channels \cite{DM-MIMO}, digital semantic communication systems \cite{mohao}, and channel estimation error mitigation \cite{zhouxingyu}. However, wireless signals not only suffer from Gaussian noise but are also highly susceptible to unaware interference due to their inherent broadcast nature, which can lead to severe performance degradation. Therefore, interference cancellation from received signals remains a critical and long-standing challenge in the wireless communication system \cite{Bameri,liufengwei}. Inspired by the success of CDDM in effectively removing Gaussian noise, it is worth investigating \textbf{whether DMs can be effectively applied to interference cancellation}.

Unlike Gaussian noise, interference patterns are vary widely and cannot be captured by a single unified distribution. As a result, the channel transition probability under interference is difficult to characterize explicitly, making it difficult to design a forward diffusion process that matches this transition. This challenge undermines the denoising paradigm introduced in CDDM, rendering it ineffective for interference. Therefore, there is an urgent need for a new, general paradigm specifically crafted to tackle interference cancellation effectively.

Motivated by these challenges, we propose an interference cancellation diffusion model (ICDM), a practical and feasible method for rapid and accurate interference cancellation. First, we present a novel and general interference cancellation paradigm, which models the interference cancellation problem as a maximum a posteriori (MAP) problem over the joint posterior probability of the signal and interference. We further theoretically prove that the solution to this problem provides excellent estimates of the signal and interference, since the joint mean squared error (MSE) between them and the original clean signal and interference has an upper bound. To solve this MAP problem, we find that Langevin dynamics \cite{Langevin} progressively approximates the solution. This method performs iterations similar to gradient descent and has been proven effective to this task \cite{lanproof}. However, applying this approach to interference cancellation still faces two key challenges: $i)$ each Langevin step requires the score function (i.e., the log-gradient), which is intractable when the interference is unaware and lacks a parametric form; $ii)$ the standard iterative scheme is inefficient: achieving high accuracy demands many steps, resulting in high computational cost and inference delay.

To address these challenges, the proposed ICDM decomposes the joint posterior probability into three factors: the prior of the desired signal, the prior of the interference, and the channel transition probability. We model the log-gradient of each prior distribution independently using a dedicated diffusion model: one capturing the log-gradient of signal, while the other learns that of interference, in line with established approaches to interference cancellation \cite{CGRSA-Net}. Next, we derive the closed-form estimates for the channel transition probability at each iteration step, which in turn yield precise evaluations of the its log-gradients. Finally, by feeding these accurately estimated log-gradients into an advanced numerical iteration method, the proposed method achieves rapid and highly accurate interference cancellation.

\begin{figure*}[t]
\centering
\includegraphics[width=0.95\linewidth]{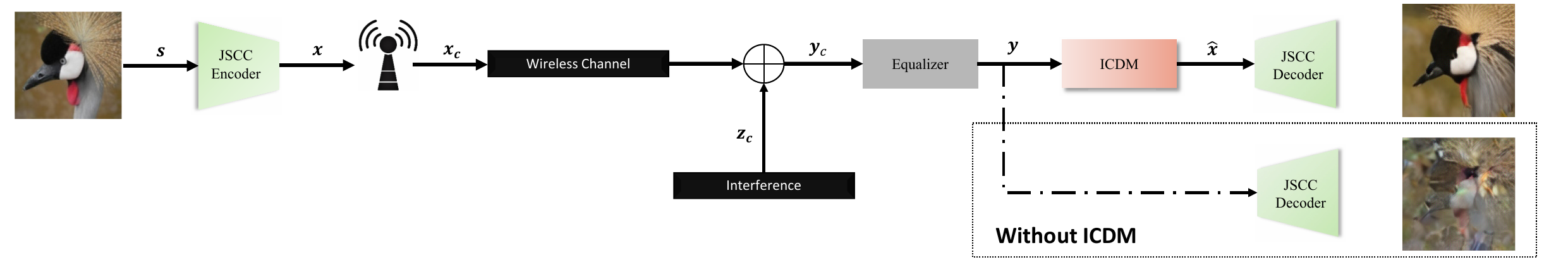}
\caption{The overall architecture of wireless semantic communications system.}
\label{ICDM_system}
\end{figure*}

The main contributions are summarized as follows.
\begin{itemize}
  \item \textbf{Interference Cancellation Paradigm:} We propose a novel and general interference cancellation paradigm, which formulates the problem as a MAP problem over the joint posterior probability of the signal and interference. We then theoretically prove that the solution to this problem yields excellent estimates of the signal and interference. Meanwhile, the Langevin dynamics—a numerically-grounded, iterative method—is proven to effectively solve this MAP problem.
  \item \textbf{Architecture of ICDM:} We design ICDM as a practical and feasible solution to this problem. The ICDM consists of two independent DMs built upon Diffusion Transformer (DiT) \cite{dit} and operates between the equalizer and the decoder. It employs a practical training strategy in which the two models independently learn the interference pattern and the signal feature, yielding log-gradients estimates of their prior distributions at each iteration for effective interference cancellation. 
  \item \textbf{Algorithm of the Methodology:} We theoretically analyze the transition probability at each iteration via Bayes' theorem and thus derive an excellent estimating expression for its log-gradient. The resulting estimation of the transition probability and the learned estimates of the prior distributions by the two DMs are combined with an advanced numerical iteration method, yielding rapid and accurate interference cancellation.
  \item \textbf{Experiments Evaluation:} We conduct extensive experiments on wireless semantic communication systems \cite{Qinzj,SwinJSCC,MambaJSCC}, demonstrating that the proposed ICDM can effectively reduce the MSE between the recovered and original transmitted signal, e,g, by $4.54$ dB, and thus achieving significant performance gains, such as a $2.47$ dB improvement in learned perceptual image patch similarity (LPIPS) \cite{lpips}. Moreover, we present a comprehensive set of studies to illustrate the advantages of ICDM over other gradient estimation methods and demonstrate the feasibility of the proposed ICDM.
\end{itemize}

The rest of this paper is organized as follows. The system model and theorem formulation are introduced in Section \ref{II}. The detail of the proposed ICDM is presented in Section \ref{III}. The extensive experimental results are presented in Section \ref{IV}. Finally, conclusions and future works are drawn in Section \ref{V}.
\section{System Model and Theorem Formulation}\label{II}
In this section, we present the system model under interference and introduce the theorem and methodology for effectively achieving interference cancellation.

\subsection{An overview of the system}
The overall architecture of the wireless semantic communications system is shown in Fig. \ref{ICDM_system}. At the transmitter, a joint source-channel coding (JSCC) encodes the source image $\mathbf{s}\in \mathbb{R}^{3\times H_s\times W_s}$ into a semantic feature vector $\mathbf{x} \in \mathbb{R}^{2k}$, which is then transformed into a complex-valued vextor $\mathbf{x_c}\in \mathbb{C}^{k}$ with normalized power by combining the preamble segment as the real part and the postamble segment as the imaginary part. The signal $\mathbf{x_c}$ is then transmitted through the wireless channel. Here $k$ represents the number of channel uses for transmiting $\mathbf{x_c}$, and $H_s$ and $W_s$ denote the height and width of the source image, respectively.

In this paper, we consider the additive white Gaussian noise (AWGN) channel and the Rayleigh fading channel. Specifically, there exists an unknown additive interference $\mathbf{z_c}$ in the wireless channel. Without loss of generality, we consider that $\mathbf{z_c}$ has the same dimension as $\mathbf{x_c}$ and continuously affects the reception of $\mathbf{x_c}$. Therefore, for $i=1,2,\cdots,k$, the $i$-th element of the received signal $\mathbf{y_c}$ is given by
\begin{align}
  y_{c,i} = \sqrt{P_{x}}h_{x,i}\ x_{c,i} + \sqrt{P_{z}}h_{z,i}\ z_{c,i} + n_{c,i},
\end{align}
where $P_{x}$ and $P_{z}$ are the transmit power of the desired signal and the interference, respectively. Under the Rayleigh fading channel, the gains $h_{x,i}$ and $h_{z,i}$ are independent and identically distributed (i.i.d) as $\mathbb{CN}(0,1)$, while both gains are unity under the AWGN channel. $n_{c,i}\sim \mathbb{CN}(0,\sigma^2)$ are i.i.d AWGN sample.

The received signal $\mathbf{y_c}$ is then equalized and represented as a real vector $\mathbf{y}\in \mathbb{R}^{2k}$. We consider that the receiver knows the channel state $\mathbf{h_x}=[h_{x,1},h_{x,2},\cdot \cdot \cdot,h_{x,k}]$ but not $\mathbf{h_z}=[h_{z,1},h_{z,2},\cdot \cdot \cdot,h_{z,k}]$. Therefore, in this paper, we apply the classic minimum mean square error (MMSE) equalizer under the Rayleigh fading channel. After equalization, ICDM is applied to $\mathbf{y}$ to cancel the interference, and the resulting signal is fed into the JSCC decoder for reconstruction as $\mathbf{\hat{s}}$.

\subsection{Theorem for Interference Cancellation}
We formulate the interference cancellation task as a MAP estimation problem. Specifically, we seek estimates $\mathbf{\hat{x}}$, $\mathbf{\hat{z}}$ that maximize the posterior probability of $\mathbf{x}$ and $\mathbf{z}$ given the observations and known channel state, i.e.,
\begin{align}\label{MAP_prb}
  (\mathbf{\hat{x}}, \mathbf{\hat{z}})= \arg\max_{\mathbf{x},\mathbf{z}} \log p_{\mathbf{x}, \mathbf{z}|\mathbf{y},\mathbf{h_x}}(\mathbf{x},\mathbf{z}|\mathbf{y},\mathbf{h_x}),
\end{align}
where $\mathbf{z}=\begin{bmatrix}Re(\mathbf{h_z}\mathbf{z_c})\\Im(\mathbf{h_z}\mathbf{z_c})\end{bmatrix}$. To illustrate the effectiveness of the formulation, we first clarify their real-valued relationship.

\begin{lemma}\label{lemma1}
The real-valued vectors $\mathbf{x}$, $\mathbf{y}$ and $\mathbf{z}$ satisfy
\begin{align}\label{realy}
  \mathbf{y} = \sqrt{P_{x}}\mathbf{W_x}\mathbf{x} + \sqrt{P_{z}}\mathbf{W_z}\mathbf{z} + \mathbf{W_n}\mathbf{n},
\end{align}
where $\mathbf{n}\sim \mathcal{N}(0,\frac{\sigma^2}{2}\mathbf{I}_{2k})$. For the Rayleigh fading channel with MMSE equalizer, $\mathbf{H}=diag(\begin{bmatrix}|\mathbf{h_x}|\\|\mathbf{h_x}|\end{bmatrix})$, $\mathbf{H_x^r}=diag(Re(\mathbf{h_x}))$, $\mathbf{H_x^I}=diag(Im(\mathbf{h_x}))$, $\mathbf{W_z}=\begin{bmatrix}\ \ \mathbf{H_x^r},\mathbf{H_x^I}\\-\mathbf{H_x^I},\mathbf{H_x^r}\end{bmatrix}(\mathbf{H}^2+\sigma^2\mathbf{I})^{-1}$,
$\mathbf{W_s}=\mathbf{H}^2(\mathbf{H}^2+\sigma^2\mathbf{I})^{-1}$, and $\mathbf{W_n}=\mathbf{H}(\mathbf{H}^2+\sigma^2\mathbf{I})^{-1}$.
For the AWGN channel, $\mathbf{W_s}=\mathbf{W_z}=\mathbf{W_n}=\mathbf{I}_{2k}$.
\end{lemma}
\begin{proof}
  For the Rayleigh fading channel with MMSE equalization, the $i$-th symbol of the equalized signal $y_{eq,i}$ is
\begin{align}
  y_{eq,i}&=\frac{h^H_{x,i}}{|h_{x,i}|^2+\sigma^2}y_{c,i}=\frac{\sqrt{P_x}h^H_{x,i}h_{x,i}}{|h_{x,i}|^2+\sigma^2}x_{c,i}\nonumber\\
  &+\frac{\sqrt{P_z}h^H_{x,i}h_{z,i}}{|h_{x,i}|^2+\sigma^2}z_{c,i}+\frac{h^H_{x,i}}{|h_{x,i}|^2+\sigma^2}n_{c,i}.
\end{align}
 
The coefficient $\frac{h^H_{x,i}h_{x,i}}{|h_{x,i}|^2+\sigma^2}=\frac{|h_{x,i}|^2}{|h_{x,i}|^2+\sigma^2}$ is a real and hence extends directly to the full vector $\mathbf{x}$, yielding the matrix $\mathbf{W_s}$. For the second term, both $h_{z,i}$ and $z_{c,i}$ are unknown, so we group them into the vector $\mathbf{z}$ and express the product $h^H_{x,i}h_{z,i}z_{c,i}$ in block-matrix form. This expression then extends naturally to the full vector $\mathbf{z}$, yielding the matrix $\mathbf{W_z}$. The third term contains the Gaussian noise sample $n_{c,i}$. By applying the resampling trick, in which the $h_{x,i}^Hn_{c,i}$ is equivalent to $|h_{x,i}|n_{c,i}$, we obtain the matrix $\mathbf{W_n}$.
 
 For the AWGN channel, the equalized signal coincides exactly with the received signal, thus we can easily derive the expression.
\end{proof}
Lemma 1 illustrates the structure of $\mathbf{z}$ in (\ref{MAP_prb}). By combining the unknown channel gains $\mathbf{h_z}$ into the interference signal $\mathbf{z_c}$, the coefficient matrices $\mathbf{W_x}$, $\mathbf{W_z}$ and $\mathbf{W_n}$ depend only on the known gains $\mathbf{h_x}$, thereby eliminating all unknown multiplication factors. Therefore, the problem reduces to one of additive interference. We then introduce two assumptions for theoretical analysis.
\begin{assumption}\label{independent_assum}
  The transmitted signal $\mathbf{x}$ and the interference $\mathbf{z}$ are independent.
\end{assumption}

\begin{assumption}\label{local_optmum}
  The functions $\log p_\mathbf{x}(\mathbf{x})$ and $\log p_\mathbf{z}(\mathbf{z})$ are each locally strongly convex, with convexity parameters  $\mu > 0$ and $\nu > 0$, respectively. Moreover, the ground truth $\mathbf{x^*}$ and $\mathbf{z^*}$, which satisfy $\mathbf{y} = \sqrt{P_{x}}\mathbf{W_x}\mathbf{x^*} + \sqrt{P_{z}}\mathbf{W_z}\mathbf{z^*} + \mathbf{W_n}\mathbf{n}$, are local optimums of $p_\mathbf{x}(\mathbf{x})$, $p_\mathbf{z}(\mathbf{z})$. 
\end{assumption}

The two assumptions are practical. In wireless communication system, it is standard to assume that the signal and interference are independent. Moreover, a practically observed data instance often appears within a region of relatively large density, which would otherwise be hard to estimate. Therefore, it is reasonable to assume that the probability density function is strong convex around the observed data, and the observed data are local optimums \cite{argmaxTheorm}. Accordingly, we have the following theorem.
\begin{theorem}
Let $\mathbf{x^*}$ and $\mathbf{z^*}$ be the ground truth satisfying $\mathbf{y} = \sqrt{P_{x}}\mathbf{W_x}\mathbf{x^*} + \sqrt{P_{z}}\mathbf{W_z}\mathbf{z^*} + \mathbf{W_n}\mathbf{n}$. The solutions $\mathbf{\hat{x}}$ and $\mathbf{\hat{z}}$ of (\ref{MAP_prb}) satisfy
  \begin{align}
    (\xi+\lambda_{min})\sqrt{||\mathbf{\hat{x}}-\mathbf{x^*}||^2+||\mathbf{\hat{z}}-\mathbf{z^*}||^2} \leq ||\mathbf{\Sigma}\mathbf{n}||,
  \end{align}
   where $\mathbf{\Sigma}=\frac{2}{\sigma^2}\mathbf{W}^T\mathbf{W}_{\mathbf{n}}^{-1}$, $\mathbf{W}=[\sqrt{P_x}\mathbf{W_x},\sqrt{P_z}\mathbf{W_z}]$, $\lambda_{min}$ is the minimum eigenvalue of $\mathbf{\Sigma }\mathbf{W}_{\mathbf{n}}^{-1}\mathbf{W}$ and $\xi =\min (\mu,\nu)$.
\end{theorem}
\begin{proof}
 The estimates $\mathbf{\hat{x}}$, $\mathbf{\hat{z}}$ that solve the MAP inference problem (\ref{MAP_prb}) satisfy
  \begin{align}\label{g_map}
    \nabla _{\mathbf{x},\mathbf{z}} \log p_{\mathbf{x}, \mathbf{z}|\mathbf{y},\mathbf{h_x}}(\mathbf{\hat{x}},\mathbf{\hat{z}}|\mathbf{y},\mathbf{h_x})=0. 
  \end{align}

For simplicity, define the concatenated vector $\mathbf{{v}}=[\mathbf{{x}};\mathbf{{z}}]$, with corresponding estimates $\mathbf{\hat{v}}=[\mathbf{\hat{x}};\mathbf{\hat{z}}]$ and ground truth $\mathbf{{v^*}}=[\mathbf{{x^*}};\mathbf{{z^*}}]$. (\ref{g_map}) can be derived as
  \begin{align}\label{local_opt}
    &\nabla _{\mathbf{v}} \log p_{\mathbf{v}|\mathbf{y},\mathbf{h_x}}(\mathbf{\hat{v}}|\mathbf{y},\mathbf{h_x})=0\nonumber\\
    \Rightarrow &\nabla _{\mathbf{v}} \log p_{\mathbf{{v}}|\mathbf{h_x}}(\mathbf{\hat{v}}|\mathbf{h_x})+\nabla _\mathbf{v}\log p_{\mathbf{y}|\mathbf{v},\mathbf{h_x}}(\mathbf{y}|\mathbf{v},\mathbf{h_x})|_{\mathbf{v}=\mathbf{\hat{v}}}=0\nonumber\\
    \Rightarrow &\nabla _{\mathbf{v}} \log p_{\mathbf{v}}(\mathbf{\hat{v}})=-\nabla _{\mathbf{v}} \log p_{\mathbf{y}|\mathbf{v},\mathbf{h_x}}(\mathbf{y}|\mathbf{v},\mathbf{h_x})|_{\mathbf{v}=\mathbf{\hat{v}}},
  \end{align}
because the prior $\mathbf{v}$ is independent of the channel gains $\mathbf{h_x}$. Specifically, from (\ref{realy}), we have
  \begin{align}\label{y-wv}
  p_{\mathbf{y}|\mathbf{v},\mathbf{h_x}}(\mathbf{y}|\mathbf{v},\mathbf{h_x}) \sim \mathcal{N}(\mathbf{y};\mathbf{Wv},\frac{\sigma^2}{2}\mathbf{W}_\mathbf{n}^2).
  \end{align}

Therefore, we have
  \begin{align}\label{grad_G}
    \nabla _{\mathbf{v}} \log p_{\mathbf{y}|\mathbf{v},\mathbf{h_x}}(\mathbf{y}|\mathbf{v},\mathbf{h_x})=\mathbf{W}^T(\frac{\sigma^2}{2}\mathbf{W}_\mathbf{n}^2)^{-1}(\mathbf{y}-\mathbf{Wv})
  \end{align}

Substituting (\ref{grad_G}) into (\ref{local_opt}), and noting that the gradients vanish at the local optimum points $\nabla_\mathbf{x} p_{\mathbf{x}}(\mathbf{x^*})=\nabla_\mathbf{z}p_{\mathbf{z}}(\mathbf{z^*})=0$, we have:
  \begin{align}
    \nabla _{\mathbf{v}} \log p_{\mathbf{v}}(\mathbf{\hat{v}})-\nabla _{\mathbf{{v}}} \log p_{\mathbf{v}}(\mathbf{{v}^*})=\mathbf{W}^T(\frac{\sigma^2}{2}\mathbf{W}_\mathbf{n}^2)^{-1}(\mathbf{y}-\mathbf{W\hat{v}}).
  \end{align}

Because of the local strong convexity of $\log p_{\mathbf{x}}(\mathbf{x})$ and $\log p_{\mathbf{z}}(\mathbf{z})$, we have
  \begin{gather}
    \langle \nabla _{\mathbf{x}} \log p_{\mathbf{x}}(\mathbf{\hat{x}})-\nabla _{\mathbf{x}} \log p_{\mathbf{x}}(\mathbf{x^*}), (\mathbf{\hat{x}}-\mathbf{x^*})\rangle 
    \geq \mu  ||\mathbf{\hat{x}}-\mathbf{x^*}||^2_2\nonumber\\
    \langle \nabla _{\mathbf{z}} \log p_{\mathbf{z}}(\mathbf{\hat{z}})-\nabla _{\mathbf{z}} \log p_{\mathbf{z}}(\mathbf{z^*}), (\mathbf{\hat{z}}-\mathbf{z^*})\rangle 
    \geq \nu  ||\mathbf{\hat{z}}-\mathbf{z^*}||^2_2,
  \end{gather}  
  where $\mu>0$, $\nu>0$. With Assumption 1, we have
  \begin{align}
    \langle \nabla _{\mathbf{v}} \log p_{\mathbf{v}}(\mathbf{\hat{v}})-\nabla _{\mathbf{v}} \log p_{\mathbf{v}}(\mathbf{v^*}), (\mathbf{\hat{v}}-\mathbf{v^*})\rangle \geq \xi ||\mathbf{\hat{v}}-\mathbf{v^*}||^2_2,
  \end{align}
  where $\xi = \min (\mu,\nu)$. Let $\mathbf{\Omega}=\mathbf{W}^T(\frac{\sigma^2}{2}\mathbf{W}_\mathbf{n}^2)^{-1}$, we derive
  \begin{align}\label{str_cov}
    &\langle\mathbf{\Omega} (\mathbf{y}-\mathbf{Wv^*}), (\mathbf{\hat{v}}-\mathbf{v^*}) \rangle \nonumber\\
    = &\langle \mathbf{\Omega} [(\mathbf{y}-\mathbf{W\hat{v}})+(\mathbf{W\hat{v}}-\mathbf{Wv^*})], (\mathbf{\hat{v}}-\mathbf{v^*}) \rangle \nonumber\\
    = &\langle \mathbf{\Omega} (\mathbf{y}-\mathbf{W\hat{v}}), (\mathbf{\hat{v}}-\mathbf{v^*}) \rangle + \langle \mathbf{\Omega W}(\mathbf{\hat{v}}-\mathbf{v^*}) ,(\mathbf{\hat{v}}-\mathbf{v^*}) \rangle \nonumber\\
    =&\langle \nabla _{\mathbf{v}} \log p_{\mathbf{v}}(\mathbf{\hat{v}})-\nabla _{\mathbf{{v}}} \log p_{\mathbf{v}}(\mathbf{v^*}), (\mathbf{\hat{v}}-\mathbf{v^*}) \rangle \nonumber\\
    &+ \langle \mathbf{\Omega W}(\mathbf{\hat{v}}-\mathbf{v^*}) ,(\mathbf{\hat{v}}-\mathbf{v^*}) \rangle \nonumber\\
    \geq& \xi  ||\mathbf{\hat{v}}-\mathbf{v^*}||^2_2 + \langle \mathbf{\Omega W}(\mathbf{\hat{v}}-\mathbf{v^*}) ,(\mathbf{\hat{v}}-\mathbf{v^*}) \rangle.
  \end{align}

Since $\mathbf{\Omega W}=\mathbf{W}^T(\frac{\sigma^2}{2}\mathbf{W}_\mathbf{n}^2)^{-1}\mathbf{W}$ is a real symmetric matrix, it has real eigenvalues. Let $\lambda_{min}$ be the minimum eigenvalue of $\mathbf{\Omega W}$. Then we have
  \begin{align}
    \langle\mathbf{\Omega} (\mathbf{y}-\mathbf{Wv^*}), (\mathbf{\hat{v}}-\mathbf{v^*}) \rangle
    \geq \xi ||\mathbf{\hat{v}}-\mathbf{v^*}||^2_2 + \lambda_{min}||\mathbf{\hat{v}}-\mathbf{v^*}||^2_2
  \end{align}
  Also, by Cauchy-schwartz inequality, we have
  \begin{align}\label{CB}
    \langle\mathbf{\Omega} (\mathbf{y}-\mathbf{Wv^*}), (\mathbf{\hat{v}}-\mathbf{v^*}) \rangle \leq ||\mathbf{\Omega}(\mathbf{y}-\mathbf{Wv^*})||\cdot ||\mathbf{\hat{v}}-\mathbf{v^*}||\nonumber\\
  \end{align}
  By combining (\ref{str_cov}) and (\ref{CB}), we have:
  \begin{gather}
    (\xi+\lambda_{min}) ||\mathbf{\hat{v}}-\mathbf{v^*}||^2_2 \leq ||\mathbf{\Omega}(\mathbf{y}-\mathbf{Wv^*})||\cdot ||\mathbf{\hat{v}}-\mathbf{v^*}||,\nonumber\\
    (\xi+\lambda_{min})||\mathbf{\hat{v}}-\mathbf{v^*}|| \leq ||\mathbf{\Omega W_n}\mathbf{n}||=||\frac{2}{\sigma^2}\mathbf{W}^T\mathbf{W}_\mathbf{n}^{-1}\mathbf{n}||.
  \end{gather}

  Finally, by expressing $\mathbf{\hat{v}}$ and $\mathbf{v^*}$ with their block elements $\mathbf{\hat{x}}$, $\mathbf{\hat{z}}$, $\mathbf{x^*}$, and $\mathbf{z^*}$, and using $\mathbf{\Sigma} =\frac{2}{\sigma^2}\mathbf{W}^T\mathbf{W}_\mathbf{n}^{-1}$, we have
  \begin{align}\label{final}
    (\xi+\lambda_{min})\sqrt{||\mathbf{\hat{x}}-\mathbf{x^*}||^2+||\mathbf{\hat{z}}-\mathbf{z^*}||^2} \leq ||\mathbf{\Sigma}\mathbf{n}||.
  \end{align}
\end{proof}

\begin{proposition}
  Under either an AWGN channel or a Rayleigh fading channel with MMSE equalization, the following inequality holds
  \begin{align}
    ||\mathbf{\hat{x}}-\mathbf{x^*}||^2+||\mathbf{\hat{z}}-\mathbf{z^*}||^2\leq \frac{1}{(\xi+\lambda_{min})^2}||\mathbf{\Sigma}\mathbf{n}||^2.
  \end{align}
\end{proposition}
\begin{proof}
  The key is to prove that $\mathbf{\Omega W}$ is a semi-positive definite matrix. It can be observed that $\mathbf{\Omega W}=\mathbf{W}^T(\frac{\sigma^2}{2}\mathbf{W}_\mathbf{n}^2)^{-1}\mathbf{W}$ is a quadratic form involving block matrices.In both case,  $\frac{\sigma^2}{2}\mathbf{W}_\mathbf{n}^2$ is a semi-positive definite matrix for any channel gain vector $\mathbf{h_{x}}$, since it is diagonal with the $i$-th diagonal element being $\frac{\sigma^2|h_{x,i}|^2}{2(|h_{x,i}|^2+\sigma^2)^2} \geq 0$ in the Rayleigh fading channel or $\frac{\sigma^2}{2}$ in the AWGN channel. Therefore, $\mathbf{\Omega W}$ is also semi-positive definite, and we have $\lambda_{min}\geq 0$. Under this condition, we have
  \begin{align}
    \xi+\lambda_{min} > 0.
  \end{align}

Therefore, (\ref{final}) can be further derived by dividing the coefficient and squaring both sides, as follows
\begin{align}
  ||\mathbf{\hat{x}}-\mathbf{x^*}||^2+||\mathbf{\hat{z}}-\mathbf{z^*}||^2\leq \frac{1}{(\xi+\lambda_{min})^2}||\mathbf{\Sigma}\mathbf{n}||^2.
\end{align}
\end{proof}

The theorem and proposition demonstrate that formulating the interference cancellation problem as a MAP problem (\ref{MAP_prb}) is appropriate. By solving the MAP problem, we can obtain the accurate estimates of the signal and interference, with the estimation error bounded by a value dependent on the signal and interference distribution, channel gain, and noise instance.

\subsection{Challenges of Solving the MAP Problem}
Now we need an algorithm to solve the MAP problem (\ref{MAP_prb}) for interference cancellation. For convenience, we denote $[\mathbf{x};\mathbf{z}]$ as $\mathbf{v}$, so $p_{\mathbf{x},\mathbf{z}|\mathbf{y},\mathbf{h_x}}(\mathbf{x},\mathbf{z}|\mathbf{y},\mathbf{h_x})$ can be written as $p_{\mathbf{v}|\mathbf{y},\mathbf{h_x}}(\mathbf{v}|\mathbf{y},\mathbf{h_x})$. Langevin dynamics \cite{Langevin} can generate samples from given distribution. With a fixed step size $\varsigma$, and an inital sample $\mathbf{\bar{v}}_0 \sim \pi (\mathbf{v})$, where $\pi$ is a prior distribution (such as the standard Gaussian distribution, from which samples can be directly drawn), the Langevin method iteratively calculates
\begin{align}
  \mathbf{\bar{v}}_t=\mathbf{\bar{v}}_{t-1}+\frac{\varsigma}{2}\nabla_{\mathbf{v}}\log p_{\mathbf{v}|\mathbf{y},\mathbf{h_x}}(\mathbf{\bar{v}}_{t-1}|\mathbf{y},\mathbf{h_x})+\sqrt{\varsigma}\mathbf{\epsilon_v},
\end{align}
where $\mathbf{\epsilon_v}\sim \mathcal{N}(0,\mathbf{I}_{4k})$. The sample $\mathbf{\bar{v}}_T$ will exactly come from $p_{\mathbf{v}|\mathbf{y},\mathbf{h_x}}(\mathbf{v}|\mathbf{y},\mathbf{h_x})$ when $\varsigma \rightarrow 0$ and $T \rightarrow \infty$. In this case, $\mathbf{v}_T$ will become an exact sample of $p_{\mathbf{v}|\mathbf{y},\mathbf{h_x}}(\mathbf{v}|\mathbf{y},\mathbf{h_x})$ and will converge to a local optimum, where $\nabla_{\mathbf{v}} \log p_{\mathbf{v}|\mathbf{y},\mathbf{h_x}}(\mathbf{\bar{v}}_T|\mathbf{y},\mathbf{h_x})=0$. Therefore, Langevin dynamics can solve the MAP problem and provide a local optimum solution.

However, the iteration method is not feasible and practical now as $ \nabla_\mathbf{v}\log p_{\mathbf{v}|\mathbf{y},\mathbf{h_x}}(\mathbf{v}|\mathbf{y},\mathbf{h_x})$ is untractable. Also, generating a high-quality sample typically requires a large number of iterations, resulting in significant computational cost and time delay \cite{ddpm,smld}, which are critical concerns in the wireless communication system. To address this, we can instead obtain the sample by solving an ordinary differential equation (ODE) with a continuous time variable $t$ as following \cite{projection}
\begin{align}\label{target_ODE}
  d\mathbf{v}=\mathbf{f}(\mathbf{v},t)dt-\frac{1}{2}g^2(t)\nabla_\mathbf{v}\log p_{\mathbf{v}|\mathbf{y},\mathbf{h_x}}(\mathbf{v}|\mathbf{y},\mathbf{h_x})dt.
\end{align}

The ODE can be solved using numerical solvers, which offer higher efficiency and accuracy compared to Langevin dynamics, as they discretize the ODE along the time variable $t$ with a step size and compute iteratively. Moreover, the joint posterior gradient $\nabla_\mathbf{v}\log p_{\mathbf{v}|\mathbf{y},\mathbf{h_x}}(\mathbf{v}|\mathbf{y},\mathbf{h_x})$ can be decomposed to enable practical and tractable estimation as 
\begin{align}\label{required_grad}
  \nabla_\mathbf{v}\log p_{\mathbf{v}|\mathbf{y},\mathbf{h_x}}(\mathbf{v}|\mathbf{y},\mathbf{h_x})&=\begin{bmatrix}
    \nabla_\mathbf{x}\log p_{\mathbf{y}|\mathbf{x},\mathbf{z},\mathbf{h_x}}(\mathbf{y}|\mathbf{x},\mathbf{z},\mathbf{h_x})\\
    \nabla_\mathbf{z}\log p_{\mathbf{y}|\mathbf{x},\mathbf{z},\mathbf{h_x}}(\mathbf{y}|\mathbf{x},\mathbf{z},\mathbf{h_x})
  \end{bmatrix}\nonumber\\
  &+\begin{bmatrix}
    \nabla_\mathbf{x}\log p_{\mathbf{x}}(\mathbf{x})\\
    \nabla_\mathbf{z}\log p_{\mathbf{z}}(\mathbf{z})
    \end{bmatrix}.
\end{align}

Therefore, to accurately and rapidly solve the MAP problem with the convertion to ODE and the decomposition of $\nabla_\mathbf{v}\log p_{\mathbf{v}|\mathbf{y},\mathbf{h_x}}(\mathbf{v}|\mathbf{y},\mathbf{h_x})$, we face the following challenges:
\begin{itemize}
  \item How to accurately estimate the gradients $\nabla_{\mathbf{x}}\log p_{\mathbf{x}}(\mathbf{x}_t)$ and $\nabla_\mathbf{z}\log p_{\mathbf{z}}(\mathbf{z}_t)$ at all time $t$, given that the prior distributions of them are both generally intractable.
  \item How to precisely derive the estimation expression of the intractable gradients $\nabla_{\mathbf{x}_t}\log p_{\mathbf{y}|\mathbf{x}_t,\mathbf{z}_t,\mathbf{h_x}}(\mathbf{y}|\mathbf{x}_t,\mathbf{z}_t,\mathbf{h_x})$ and $\nabla_{\mathbf{z}_t}\log p_{\mathbf{y}|\mathbf{x}_t,\mathbf{z}_t,\mathbf{h_x}}(\mathbf{y}|\mathbf{x}_t,\mathbf{z}_t,\mathbf{h_x})$ across all time steps $t$, which are essential for guiding the sampling process.
  \item How to effectively combine these gradients using advanced numerical solvers to integrate the ODE, enabling accurate and fast recovery of the clean signal.
\end{itemize}

\section{Interference Cancellation Diffusion Models}\label{III}
\begin{figure}[t]
  \centering
  \includegraphics[width=0.97\linewidth]{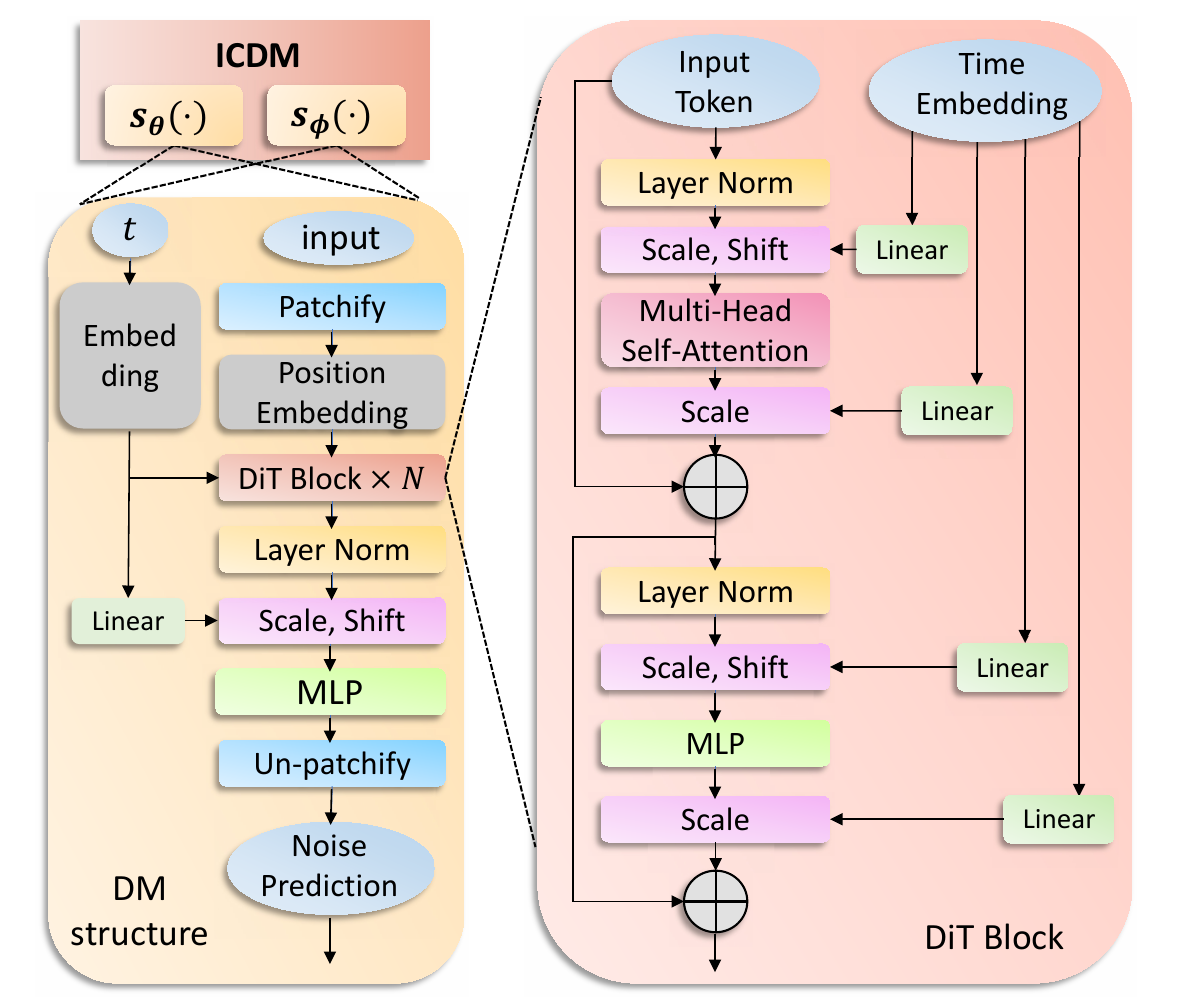}
  \caption{The overall structure of the proposed ICDM.}
  \label{ICDM_structure}
  \end{figure}
In this section, we present the gradients estimation method and advanced numerical iteration sampling algorithms used in the proposed ICDM to solve the ODE efficiently and accurately, enabling the recovery of the final clean signal.

\subsection{Training Algorithm of ICDM for Prior Gradient Estimation}
The ICDM consists of two diffusion models, $\mathbf{s_\theta}(\cdot,t)$ and $\mathbf{s_\phi}(\cdot,t)$, which are trained to accurately estimate the gradients $\nabla_\mathbf{x}\log p_{\mathbf{x}}(\mathbf{x})$ and $\nabla_\mathbf{z}\log p_{\mathbf{z}}(\mathbf{z})$ for all $t$. The structure of ICDM based on DiT \cite{dit} is shown in Fig. \ref{ICDM_structure}.
We first add noise to the data to create a sequence of discrete variables $\mathbf{x}_t$ and $\mathbf{z}_t$ for all $t\in[0,T]$. Specifically, given samples $\mathbf{x}_T\sim p_\mathbf{x}(\mathbf{x})$ and $\mathbf{z}_T\sim p_\mathbf{z}(\mathbf{z})$, the variables $\mathbf{x}_t$ and $\mathbf{z}_t$ are generated as follows
\begin{gather}\label{diff_process}
  \mathbf{x}_t=\sqrt{\alpha_t}\mathbf{x}_T+\sqrt{1-\alpha_t}\mathbf{\epsilon_x}, \quad \mathbf{z}_t=\sqrt{\alpha_t}\mathbf{z}_T+\sqrt{1-\alpha_t}\mathbf{\epsilon_z},
\end{gather}
where $\alpha_t$ denotes the noise schedule, and $\mathbf{\epsilon_x}, \mathbf{\epsilon_z}\sim \mathcal{N}(0,\mathbf{I}_{2k})$ represent the Gaussian noise. The true distributions $q(\mathbf{x}_t|\mathbf{x}_T)$ and $q(\mathbf{z}_t|\mathbf{z}_T)$ can be expressed as 
\begin{align}
  q(\mathbf{x}_t|\mathbf{x}_T)=\mathcal{N}(\mathbf{x}_t;\sqrt{\alpha_t}\mathbf{x}_T,(1-\alpha_t)\mathbf{I}_{2k}),\nonumber\\
  q(\mathbf{z}_t|\mathbf{z}_T)=\mathcal{N}(\mathbf{z}_t;\sqrt{\alpha_t}\mathbf{z}_T,(1-\alpha_t)\mathbf{I}_{2k}).
\end{align}

Then, we train two diffusion models $\mathbf{s_\theta}(\cdot,t)$ and $\mathbf{s_\phi}(\cdot,t)$, with the training objectives of minimizing their negative log likehood functions as following
\begin{gather}
  \min _\mathbf{\theta} \mathbb{E}_{p_\mathbf{x}(\mathbf{x}_T)} \log(p_\mathbf{\theta}(\mathbf{x}_T)), \quad
  \min_\mathbf{\phi}\mathbb{E}_{p_\mathbf{z}(\mathbf{z}_T)} \log(p_\mathbf{\phi}(\mathbf{z}_T)).
\end{gather}

With the variational bound and the diffusion process formulation \cite{ddpm}, the training objective can be derived to noise prediction, leading to the following loss functions for optimizing parameters $\mathbf{\theta}$ and $\mathbf{\phi}$
\begin{gather}\label{sqloss}
  \mathcal{L}_{\mathbf{x}}(\mathbf{\theta})=\mathbb{E}_{t,\mathbf{x}_T,\mathbf{\epsilon_x}}\left[\left\|\mathbf{s_\theta}(\mathbf{x}_t,t)-\mathbf{\epsilon_x}\right\|^2_2\right],\nonumber\\
\mathcal{L}_{\mathbf{z}}(\mathbf{\phi})=\mathbb{E}_{t,\mathbf{z}_T,\mathbf{\epsilon_z}}\left[\left\|\mathbf{s_\phi}(\mathbf{z}_t,t)-\mathbf{\epsilon_z}\right\|^2_2\right].
\end{gather}

 Notably, the training process can be implemented in interference environment beacuse the diffusion models $\mathbf{s_\phi}(\cdot, t)$ and $\mathbf{s_\theta}(\cdot,t)$ are trained independently of each other, eliminating the need for jointly training with the transmitter. The training algorithm for $\mathbf{s}_\phi(\cdot,t)$ in ICDM is outlined in Algorithm \ref{alg:training}, and the procedure for training $\mathbf{s}_\theta(\cdot,t)$ follows a similar approach.
\begin{algorithm}[t]
  \caption{Training Algorithm of $\mathbf{s_\phi}(\cdot,t)$ in ICDM}
  \label{alg:training}
  {\bf {Input:}}
	\small{Collection of interference samples $\mathbf{Z}$, $T$, noise schedule $\alpha_t$.} \\
	{\bf {Output:}}
	\small{Well trained $\mathbf{s_\phi}(\cdot,t)$.}
  \begin{algorithmic}[1]
    \REPEAT
    \STATE $\mathbf{z}_T \sim \mathbf{Z}$ 
    \STATE $\mathbf{\epsilon_z}\sim\mathcal{N}(0,\mathbf{I}_{2k})$
    \STATE $t \sim Uniform(0,T)$
    \STATE Take gradient descent step on \\
    $\nabla_\phi ||\mathbf{s}_\phi(\sqrt{\alpha_t}\mathbf{z}_T+\sqrt{1-\alpha_t}\mathbf{\epsilon_z},t),\mathbf{\epsilon}_z||^2_2$
    \UNTIL {converged}
  \end{algorithmic}
  \end{algorithm}

Using Tweedie's formula, we derive the true expectations $\mathbf{\mu}_{q(\mathbf{x}_t|\mathbf{x}_T)}$ and $\mathbf{\mu}_{q(\mathbf{z}_t|\mathbf{z}_T)}$ of $q(\mathbf{x}_t|\mathbf{x}_T)$ and $q(\mathbf{z}_t|\mathbf{z}_T)$ from their observed samples $\mathbf{x}_t$ and $\mathbf{z}_t$
\begin{align}\label{Tweedie}
  &\mathbb{E}[\mathbf{\mu}_{q(\mathbf{x}_t|\mathbf{x}_T)}|\mathbf{x}_t]=\mathbf{x}_t+(1-\alpha_t)\nabla_{\mathbf{x}}\log p_{\mathbf{x}}(\mathbf{x}_t),\nonumber\\
  &\mathbb{E}[\mathbf{\mu}_{q(\mathbf{z}_t|\mathbf{z}_T)}|\mathbf{z}_t]=\mathbf{z}_t+(1-\alpha_t)\nabla_{\mathbf{z}}\log p_{\mathbf{z}}(\mathbf{z}_t).
\end{align}
By comparing (\ref{Tweedie}) with the diffusion process (\ref{diff_process}), we derive
\begin{align}
  \nabla_{\mathbf{x}}\log p_{\mathbf{x}}(\mathbf{x}_t)=-\frac{1}{\sqrt{1-\alpha_t}}\mathbf{\epsilon_x},\nonumber\\
  \nabla_{\mathbf{z}}\log p_{\mathbf{z}}(\mathbf{z}_t)=-\frac{1}{\sqrt{1-\alpha_t}}\mathbf{\epsilon_z}.
\end{align}

The derived expressions indicate that the noise predicted by the trained models corresponds to the negative log-gradients. Furthermore, \cite{prior_KL} establishes the accuracy of this estimation by demonstrating that the Kullback-Leibler (KL) divergence between the true distribution and the estimated distribution is bounded above by a small value. Therefore, with the trained models $\mathbf{s_\theta}(\cdot,t)$ and $\mathbf{s_\phi}(\cdot,t)$, we can accurately estimate the gradients $\nabla_{\mathbf{x}}\log p_{\mathbf{x}}(\mathbf{x})$ and $\nabla_{\mathbf{z}}\log p_{\mathbf{z}}(\mathbf{z})$ at any time step $t$. This enables precise solutions to the ODE through advanced discrete numerical methods. 

\subsection{Deriving Expression for Guidance Gradient Estimation}
To derive the estimation expressions for the log-gradients $\nabla_\mathbf{x}\log p_{\mathbf{y}|\mathbf{x},\mathbf{z},\mathbf{h_x}}(\mathbf{y}|\mathbf{x},\mathbf{z},\mathbf{h_x})$ and $\nabla_\mathbf{z}\log p_{\mathbf{y}|\mathbf{x},\mathbf{z},\mathbf{h_x}}(\mathbf{y}|\mathbf{x},\mathbf{z},\mathbf{h_x})$, we first define $\mathbf{v}$ as in (\ref{required_grad}) and adopt $\mathbf{v}_t=[\mathbf{x}_t;\mathbf{z}_t]$ for notational simplicity. The gradient $\nabla_{\mathbf{v}_t}\log p_{\mathbf{y}|\mathbf{v}_t,\mathbf{h_x}}(\mathbf{y}|\mathbf{v}_t,\mathbf{h_x})$ is intractable \cite{dps,gdm,rmstruct}, as the iterative computation process only provides the noise observation $\mathbf{v}_t$ at time step $t$, while the ground truth $\mathbf{v}_T$ satisfying the measurement equation $\mathbf{y} = \mathbf{W}\mathbf{v}_T + \mathbf{W_n}\mathbf{n}$ remains inaccessible, where $\mathbf{W}=[\sqrt{P_x}\mathbf{W_x},\sqrt{P_z}\mathbf{W_z}]$. To accurately estimate the gradient, we first derive its probability density function as
\begin{gather}
  p_{\mathbf{y}|\mathbf{v}_t,\mathbf{h_x}}(\mathbf{y}|\mathbf{v}_t,\mathbf{h_x})\nonumber\\
  =\int p_{\mathbf{v}_T|\mathbf{v}_t}(\mathbf{v}_T|\mathbf{v}_t)p_{\mathbf{y}|\mathbf{v}_T,\mathbf{h_x}}(\mathbf{y}|\mathbf{v}_T,\mathbf{h_x})d\mathbf{v}_T.
\end{gather}

This is because the processes of sampling $\mathbf{v}_t$ from $\mathbf{x}_T$ and sampling $\mathbf{y}$ from $\mathbf{v}_T$ are independent. Here, the conditional distribution $p_{\mathbf{y}|\mathbf{v}_T,\mathbf{h_x}}(\mathbf{y}|\mathbf{v}_T,\mathbf{h_x}) = \mathcal{N}(\mathbf{y};\mathbf{Wv}_T,\frac{\sigma^2}{2}\mathbf{W}_\mathbf{n}^2)$ is tractable since $\mathbf{v}_T$ represents the ground truth. Therefore, we analyze the posterior distribution $p_{\mathbf{v}_T|\mathbf{v}_t,\mathbf{h_x}}(\mathbf{v}_T|\mathbf{v}_t,\mathbf{h_x})$. Using Bayes' theorem, we decompose the distribution as follows
\begin{align}
  p_{\mathbf{v}_T|\mathbf{v}_t}(\mathbf{v}_T|\mathbf{v}_t)\varpropto p_{\mathbf{v}_t|\mathbf{v}_T}(\mathbf{v}_t|\mathbf{v}_T)p_{\mathbf{v}_T}(\mathbf{v}_T).
\end{align}

It can be observed that the distribution depends on the prior distribution $p_{\mathbf{v}_T}(\mathbf{v}_T)$. A specific and practical distribution function should be assumed because although the prior distribution is modeled by the proposed ICDM using two diffusion models, it can only sample from the distribution and cannot express it with close-form expression, which is required to obtain the distribution of $p_{\mathbf{y}|\mathbf{v}_t,\mathbf{h_x}}(\mathbf{y}|\mathbf{v}_t,\mathbf{h_x})$. Turning to the interference cancellation problem, both the transmitted signal and the interference are encoded signals, which are expected to follow a Gaussian distribution to maximize the transmission rate. Therefore, although their prior distributions may be complex, we assume here that $p_{\mathbf{v}_T}(\mathbf{v}_T)=\mathcal{N}(0,\hat{\sigma}^2)$ for the purpose of estimation. Under this assumption, we have
\begin{align}
  p_{\mathbf{v}_T|\mathbf{v}_t}(\mathbf{v}_T|\mathbf{v}_t)=\mathcal{N}(\mathbf{v}_T;\frac{\hat{\sigma}^2\mathbf{v}_t}{(1-\alpha_t)+\hat{\sigma}^2},\frac{(1-\alpha_t)\hat{\sigma}^2}{(1-\alpha_t)+\hat{\sigma}^2}\mathbf{I}_{2k}).
\end{align}

By multiplying it with $p_{\mathbf{y}|\mathbf{v}_T,\mathbf{h_x}}(\mathbf{y}|\mathbf{v}_T,\mathbf{h_x})$, we can derive 
\begin{align}
  &p_{\mathbf{v}_T,\mathbf{y}|\mathbf{v}_t,\mathbf{h_x}}(\mathbf{v}_T,\mathbf{y}|\mathbf{v}_t,\mathbf{h_x})=\mathcal{N}(\mathbf{v}_T,\mathbf{y};\zeta_t\mathbf{\bar{\mu}}_t ,\mathbf{\bar{\Theta}}_t ),\nonumber\\
  &\mathbf{\zeta}_t=\frac{\hat{\sigma}^2}{(1-\alpha_t)+\hat{\sigma}^2},\quad \quad \quad \quad \quad \ \mathbf{\bar{\mu}}_t=\begin{bmatrix}
    \mathbf{v}_t\\
    \mathbf{W}\mathbf{v}_t
    \end{bmatrix},\nonumber\\
    &\mathbf{\bar{\Theta}}_t=\frac{(1-\alpha_t)\hat{\sigma}^2}{(1-\alpha_t)+\hat{\sigma}^2}
    \begin{bmatrix}
    \mathbf{I}_{4k} & \mathbf{W}^T\\
    \mathbf{W} & \frac{\sigma^2}{2}\frac{(1-\alpha_t)+\hat{\sigma}^2}{(1-\alpha_t)\hat{\sigma}^2}\mathbf{W}_\mathbf{n}^2+\mathbf{W}\mathbf{W}^T
    \end{bmatrix}.
\end{align}
Therefore, we can obtain the marginal distribution 
\begin{align}\label{gradient}
  &p_{\mathbf{y}|\mathbf{v}_t,\mathbf{h_x}}(\mathbf{y}|\mathbf{v}_t,\mathbf{h_x})=\mathcal{N}(\mathbf{y}; \zeta_t\mathbf{Wv}_t,\mathbf{\Theta}_t),\nonumber\\
  &\mathbf{\Theta}_t=\frac{\sigma^2}{2}\mathbf{W}_\mathbf{n}^2+\frac{(1-\alpha_t)\hat{\sigma}^2}{(1-\alpha_t)+\hat{\sigma}^2}\mathbf{W}\mathbf{W}^T.
\end{align}

Now we can compute the gradient
\begin{align}
\nabla_{\mathbf{v}_t}\log p_{\mathbf{y}|\mathbf{v}_t,\mathbf{h_x}}(\mathbf{y}|\mathbf{v}_t,\mathbf{h_x})=\zeta_t^2\mathbf{W}^T\mathbf{\Theta}_t^{-1}(\frac{\mathbf{y}}{\zeta_t}-\mathbf{Wv}_t).
\end{align}

It is worth noting that we derive the gradient under the assumption that the prior distribution $p_{\mathbf{v}_T}(\mathbf{v}_T)$ is a Gaussian. However, the true prior is likely more complex, even if it closely approximates a Gaussian. Therefore, we retain the form of the gradient (\ref{gradient}) and approximate the parameters $\zeta_t$ and $\mathbf{\Theta}_t$ using surrogate variables that exhibits similar behavior. Focusing first on $\mathbf{\Theta}_t$, one can show that $\mathbf{WW}^T=\mathbf{W_xW}_\mathbf{x}^T+\mathbf{W_zW}_\mathbf{z}^T$ is diagonal, with its $i$-th and $(i+k)$-th diagonal elements given by $\frac{|h_{x,i}|^4+|h_{x,i}|^2}{(|h_{x,i}|^2+\sigma^2)^2}$. Moreover, the effect of the channel gain is reflected in $\mathbf{W}^T$. Therefore, we ignore this term and instead use an appropriate estimate
\begin{align}
  \mathbf{\Theta}_t\approx\frac{\sigma^2}{2}\mathbf{W}_\mathbf{n}^2.
\end{align}

Next, concentrating on $\zeta_t$, we estimate it directly using the expression with $\hat{\sigma}^2=1$ as
\begin{align}
  \zeta_t\approx\frac{1}{2-\alpha_t}.
\end{align}
\begin{figure*}[t]
  \centering
  \includegraphics[width=0.95\linewidth]{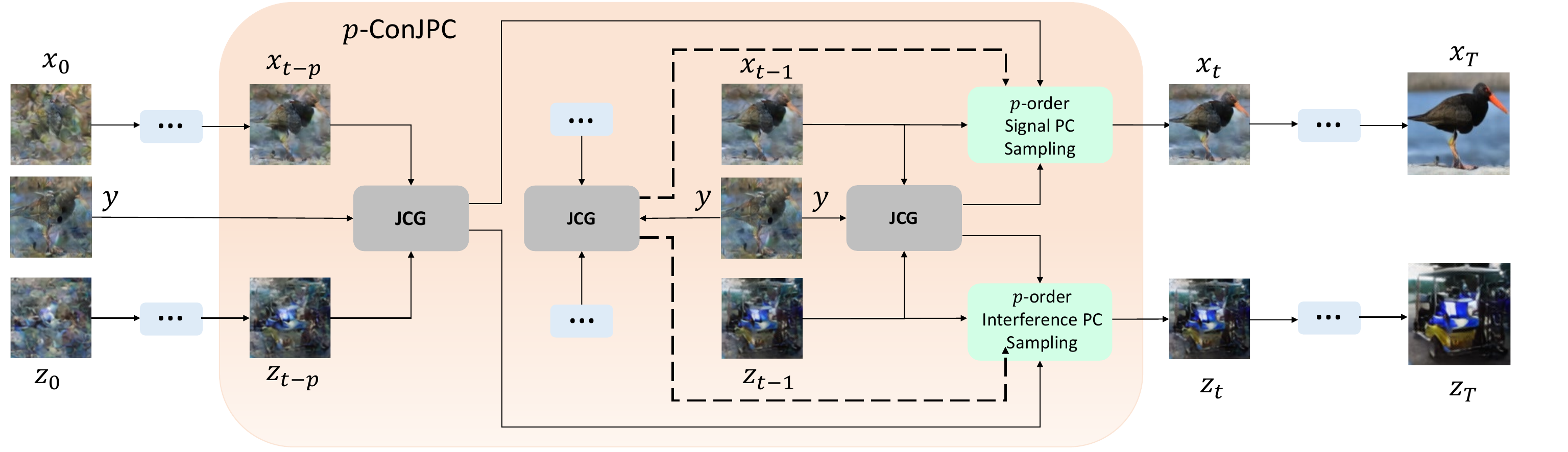}
  \caption{The overall sampling process of ICDM. $\mathbf{x}_t$ and $\mathbf{z}_t$ are invisible signals and we decode them into images here for illustration.}
  \label{p-ConJPC}
  \end{figure*}
  \begin{figure}[t]
    \centering
    \includegraphics[width=0.95\linewidth]{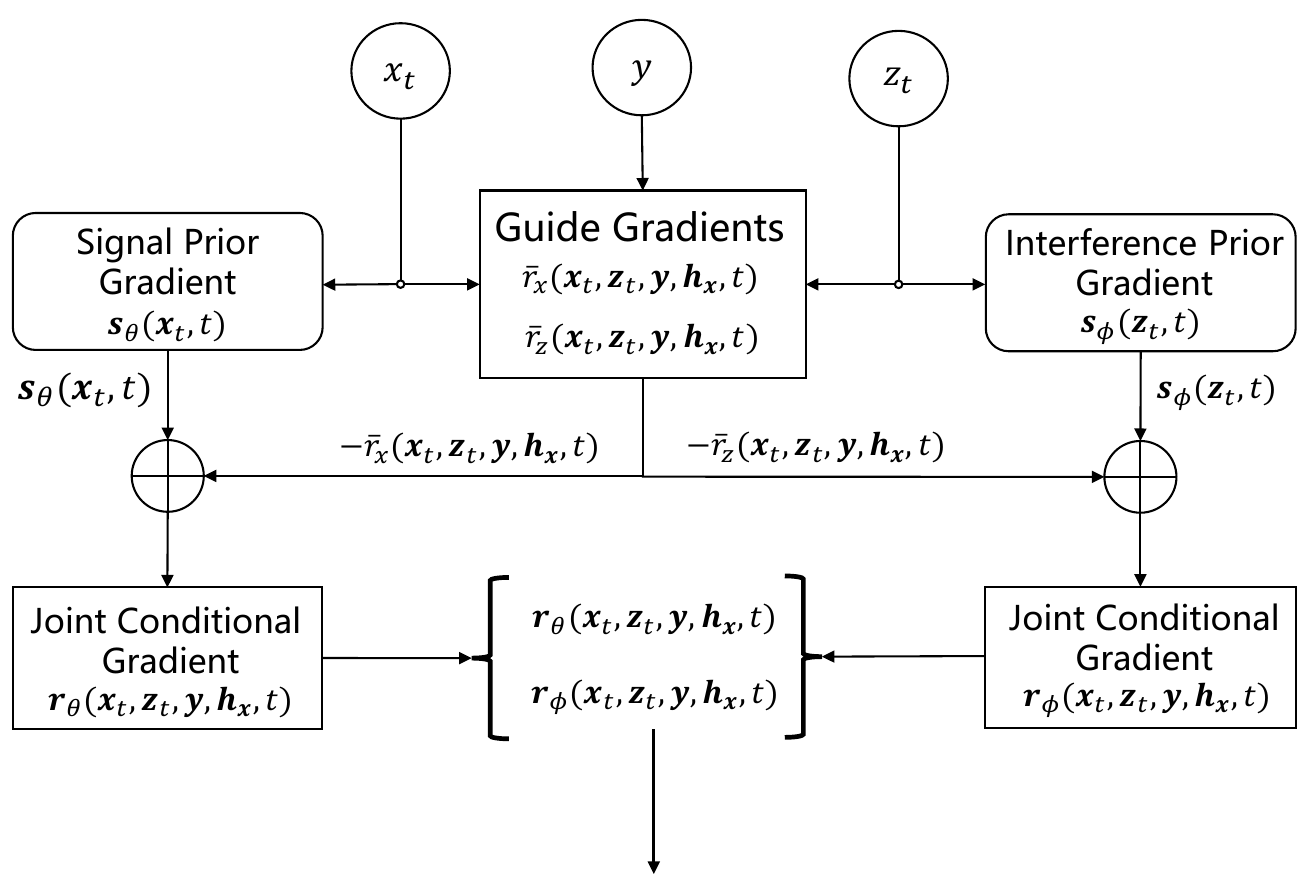}
    \caption{The illustration of the JCG module.}
    \label{JCG}
  \end{figure}
Additionally, the related constants are ignored. We obtain the gradients and rewrite them in terms of $\mathbf{x}$ and $\mathbf{z}$ as
\begin{align}\label{Ours}
    &\nabla_{\mathbf{x}_t}\log p_{\mathbf{y}|\mathbf{x}_t,\mathbf{z}_t,\mathbf{h_x}}(\mathbf{y}|\mathbf{x}_t,\mathbf{z}_t,\mathbf{h_x})\approx \mathbf{\bar{r}}_x(\mathbf{x}_t,\mathbf{z}_t,\mathbf{y},\mathbf{h_x})\nonumber\\
    &=\sqrt{P_x}\zeta_t^2\mathbf{W}_\mathbf{x}^T(\mathbf{W}_{\mathbf{n}}^2)^{-1}(\frac{\mathbf{y}}{\zeta_t}-\sqrt{P_x}\mathbf{W_x}\mathbf{x}_t-\sqrt{P_z}\mathbf{W_z}\mathbf{z}_t),\nonumber\\
    &\nabla_{\mathbf{z}_t}\log p_{\mathbf{y}|\mathbf{x}_t,\mathbf{z}_t,\mathbf{h_x}}(\mathbf{y}|\mathbf{x}_t,\mathbf{z}_t,\mathbf{h_x})\approx \mathbf{\bar{r}}_z(\mathbf{x}_t,\mathbf{z}_t,\mathbf{y},\mathbf{h_x})\nonumber\\
    &=\sqrt{P_z}\zeta_t^2\mathbf{W}_\mathbf{z}^T(\mathbf{W}_{\mathbf{n}}^2)^{-1}(\frac{\mathbf{y}}{\zeta_t}-\sqrt{P_x}\mathbf{W_x}\mathbf{x}_t-\sqrt{P_z}\mathbf{W_z}\mathbf{z}_t).
\end{align}

These precisely estimated gradients not only improve the accuracy of interference cancellation but also enable the proposed ICDM to adaptively adjust the sampling process based on different channel gains, noise schedules, and the sample steps, thereby enhancing its adaptability to various scenarios.
\subsection{Interference Cancellation Algorithm of ICDM}

With these estimated gradients, we develop the ICDM interference cancellation algorithm by solving the ODE in (\ref{target_ODE}) with the proposed \textit{$p$-order joint conditional predict-correct ($p$-ConJPC) sampler} as shown in Fig. \ref{p-ConJPC}. The $p$-ConJPC is a specially designed iteration method that quickly and accurately estimates the clean signal $\mathbf{\hat{x}}$ and interference $\mathbf{\hat{z}}$ from the initial samples $\mathbf{x}_0$ and $\mathbf{z}_0$ drawn from $\mathcal{N}(0,\mathbf{I})$. It does so by combining the new gradients with the advanced UniPC method \cite{unipc}, which achieves $(p+1)$-th order of accuracy.

Using the definition in (\ref{diff_process}), the ODE in (\ref{target_ODE}) is specifically decomposed into two related ODEs for $\mathbf{x}_t$ and $\mathbf{z}_t$, namely
\begin{align}\label{target_ODE2}
  \frac{d\mathbf{x}_t}{dt}=f(t)\mathbf{x}_t-\frac{1}{2}g^2(t)\nabla _{\mathbf{x}_t}\log p(\mathbf{x}_t,\mathbf{z}_t|\mathbf{y},\mathbf{h_x}),\\
  \frac{d\mathbf{z}_t}{dt}=f(t)\mathbf{z}_t-\frac{1}{2}g^2(t)\nabla _{\mathbf{z}_t}\log p(\mathbf{x}_t,\mathbf{z}_t|\mathbf{y},\mathbf{h_x}),\nonumber
\end{align}
where $f(t)=\frac{d \log \sqrt{\alpha_t}}{dt}$, $g^2(t)=\frac{d(1-\alpha_t)}{dt}-2(1-\alpha_t)\frac{d\log \sqrt{\alpha_t}}{dt}$. 

The proposed $p$-order ConJPC sampler consists of three components: $1)$ a joint conditional gradient (JCG) module as illustrated in Fig. \ref{JCG}; $2)$ a $p$-order signal predict-correct (PC) sampling module, as described in Algorithm. \ref{alg:signal-sampling}; $3)$ a $p$-order interference PC sampling module. Given the previous $p$ samples $\{\mathbf{x}_{t-m}\}^p_{m=1}$ and $\{\mathbf{z}_{t-m}\}^p_{m=1}$, the $p$ JCG modules compute $2p$ joint conditional gradients $\{\mathbf{r}_\theta(\mathbf{x}_{t-m},\mathbf{z}_{t-m},\mathbf{y},\mathbf{h_x},t-m)\}^p_{m=1}$ and $\{\mathbf{r}_\phi(\mathbf{x}_{t-m},\mathbf{z}_{t-m},\mathbf{y},\mathbf{h_x},t-m)\}^p_{m=1}$, which are then fed, respectively, into the $p$-order signal PC sampling module and the $p$-order interference PC sampling module. Here, we have
\begin{align}
  \mathbf{r}_\theta(\mathbf{x}_{t},\mathbf{z}_{t},\mathbf{y},\mathbf{h_x},t)&=\mathbf{s}_\mathbf{\theta}(\mathbf{x}_{t},t)-\beta \mathbf{\bar{r}}_x(\mathbf{x}_t,\mathbf{z}_t,\mathbf{y},\mathbf{h_x})\nonumber\\
  &\approx \nabla _{\mathbf{x}_t}\log p(\mathbf{x}_t,\mathbf{z}_t|\mathbf{y},\mathbf{h_x}),\nonumber\\
   \mathbf{r}_\phi(\mathbf{x}_{t},\mathbf{z}_{t},\mathbf{y},\mathbf{h_x},t)&=\mathbf{s}_\mathbf{\phi}(\mathbf{z}_{t},t)-\gamma \mathbf{\bar{r}}_z(\mathbf{x}_t,\mathbf{z}_t,\mathbf{y},\mathbf{h_x})\nonumber\\
  &\approx \nabla _{\mathbf{z}_t}\log p(\mathbf{x}_t,\mathbf{z}_t|\mathbf{y},\mathbf{h_x}).
\end{align}

\begin{remark}\label{hy-par}
 The parameters $\beta$ and $\gamma$ combine all the constants in the estimated gradients and serve as guiding intensities, as introduced in diffusion models with classifier guidance \cite{diffbeatgan}. Here, a large value means the reconstructed signal is closer to the received signal, while a small value means the reconstructed signal is more likely to be generated. Therefore, these parameters balance consistency with the received signal and overall realism in interference cancellation.
\end{remark}

The two sampling modules combine the gradients from the $p$ JCG modules and use a predicter to generate the samples
\begin{align}\label{sampler1}
    &\mathbf{\bar{x}}_{t}=-\sqrt{1-\alpha_t}(e^{\eta _t}-1)\mathbf{r}_\mathbf{\theta}(\mathbf{x}_{t-1},\mathbf{z}_{t-1},\mathbf{y},\mathbf{h_x},t-1)\nonumber\\
    &+\sqrt{\frac{\alpha_t}{\alpha_{t-1}}}\mathbf{x}_{t-1}-\sqrt{1-\alpha_t}\eta _t \sum_{m=1}^{p-1}\frac{o^{p-1}_{m,t}}{w_{m,t}}\mathbf{D}^x_{m,t},\\
    &\mathbf{\bar{z}}_{t}=-\sqrt{1-\alpha_t}(e^{\eta _t}-1)\mathbf{r}_\mathbf{\theta}(\mathbf{z}_{t-1},\mathbf{z}_{t-1},\mathbf{y},\mathbf{h_x},t-1)\nonumber\\
    &+\sqrt{\frac{\alpha_t}{\alpha_{t-1}}}\mathbf{z}_{t-1}-\sqrt{1-\alpha_t}\eta _t \sum_{m=1}^{p-1}\frac{o^{p-1}_{m,t}}{w_{m,t}}\mathbf{D}^z_{m,t},\nonumber\\
  \end{align}
where $\rho_t=\frac{1}{2}\log \frac{\alpha_t}{1-\alpha_t}$, $\eta_t=\rho _t-\rho _{t-1}$, $w_{m,t}=\frac{\rho _{t-m-1}-\rho _{t-1}}{\eta_t}$, for $m=1,2,\cdots,p-1$ and $w_{p,t}=1$. Let $\mathbf{o}^p_t=[o^{p}_{1,t},o^{p}_{2,t},\cdot,\cdot,\cdot,o^{p}_{p,t}]$, which is now can be computed as \cite{unipc}
\begin{align}
  \mathbf{o}^{p}_t=\frac{1}{\eta_t}\mathbf{\Gamma }^{-1}_{p}(t)\mathbf{b}_p(t).
\end{align}
Here, we have
\begin{align}\label{PC-parameter}
  \mathbf{\Gamma }_{p}(t)=\begin{bmatrix}
    1 & 1 &\cdot \cdot \cdot &1\\
    w_{1,t} & w_{2,t}& \cdot \cdot \cdot &w_{p,t}\\
    \cdot \cdot \cdot & \cdot \cdot \cdot &\cdot \cdot \cdot &\cdot \cdot \cdot\\
    w^{p-1}_{1,t} & w^{p-1}_{2,t} &\cdot \cdot \cdot &w^{p-1}_{p,t}
    \end{bmatrix},
\end{align}
\begin{gather}
  \mathbf{b}_p(t)=[(\frac{\psi_1-\frac{1}{1!}}{\eta_t})\cdot 1!,(\frac{\psi_2 -\frac{1}{2!}}{\eta_t})\cdot 2!,\cdots,(\frac{\psi_p-\frac{1}{p!}}{\eta_t})\cdot p!]^T\nonumber\\
  \psi_1=\frac{e^{\eta_t}-1}{\eta_t},\quad \psi_2=\frac{e^{\psi_1}-1}{\eta_t},\cdots,\psi_p=\frac{e^{\psi_{p-1}}-1}{\eta_t}.\nonumber
\end{gather}
We obtain $o^{p-1}_{m,t}$ by converting the $p$-order equations to $(p-1)$-order. The matrices $\mathbf{D}^x_{m,t}$ and $\mathbf{D}^z_{m,t}$ for $m=1,2,\cdots,p-1$, are computed from the gradients as following
\begin{align}
  {w_{m,t}}\mathbf{D}^x_{m,t}&=\mathbf{r}_\theta(\mathbf{x}_{t-m-1},\mathbf{z}_{t-m-1},\mathbf{y},\mathbf{h_x},t-m-1)\nonumber\\
  &-\mathbf{r}_\theta(\mathbf{x}_{t-1},\mathbf{z}_{t-1},\mathbf{y},\mathbf{h_x},t-1),\\
  {w_{m,t}}\mathbf{D}^z_{m,t}&=\mathbf{r}_\phi(\mathbf{x}_{t-m-1},\mathbf{z}_{t-m-1},\mathbf{y},\mathbf{h_x},t-m-1)\nonumber\\
  &-\mathbf{r}_\phi(\mathbf{x}_{t-1},\mathbf{z}_{t-1},\mathbf{y},\mathbf{h_x},t-1).
\end{align}

The method employs generalized polynomial interpolation to determine the optimal coefficients for solving the ODE (\ref{target_ODE2}) with $p$-th order of accuracy.

After predicting the samples $\mathbf{\bar{x}}_{t}$ and $\mathbf{\bar{z}}_{t}$, the corrector further refines them to obtain more accurate samples at the current time step with $(p+1)$-th order accuracy. Specifically, the predicted samples are used to compute the gradients via $\mathbf{s}_\theta(\mathbf{\bar{x}}_t,t)$ and $\mathbf{z}_\theta(\mathbf{\bar{z}}_t,t)$, which are then utilized to calculate $\mathbf{D}^x_{p,t}$ and $\mathbf{D}^x_{p,t}$ as
\begin{align}
  \mathbf{D}^x_{p,t}=\frac{1}{w_{p,t}}\left[\mathbf{s}_\theta(\mathbf{\bar{x}}_{t},t)-\mathbf{r}_\theta(\mathbf{x}_{t-1},\mathbf{z}_{t-1},\mathbf{y},\mathbf{h_x},t-1)\right],\nonumber\\
  \mathbf{D}^z_{p,t}=\frac{1}{w_{p,t}}\left[\mathbf{s}_\phi(\mathbf{\bar{z}}_{t},t)-\mathbf{r}_\phi(\mathbf{x}_{t-1},\mathbf{z}_{t-1},\mathbf{y},\mathbf{h_x},t-1)\right].
\end{align}

\begin{algorithm}[t]
  \caption{$p$-order signal PC sampling algorithm}
  \label{alg:signal-sampling}
  {\bf {Input:}}
	\small{$\mathbf{x}_{t-1}$, $\{\mathbf{r}_\theta(\mathbf{x}_{t-m},\mathbf{z}_{t-m},\mathbf{y},\mathbf{h_x},t-m)\}^p_{m=1}$}, $t$, $\alpha_t$. \\
	{\bf {Output:}}
	\small{$\mathbf{x}_t$.}
  \begin{algorithmic}[1]
    \STATE  $\rho_t=\frac{1}{2}\log \frac{\alpha_t}{1-\alpha_t}$, $\eta_t=\rho _t-\rho _{t-1}$
    \STATE \textbf{Predicter}:
    \FOR {$m=1$ to $p-1$}
    \STATE $w_{m,t}=\frac{1}{\eta_t}(\rho _{t-m-1}-\rho _{t-1})$
    \STATE ${w_{m,t}}\mathbf{D}^x_{m,t}=\mathbf{r}_\theta(\mathbf{x}_{t-m-1},\mathbf{z}_{t-m-1},\mathbf{y},\mathbf{h_x},t-m-1)$\\
    $-\mathbf{r}_\theta(\mathbf{x}_{t-1},\mathbf{z}_{t-1},\mathbf{y},\mathbf{h_x},t-1)$
    \ENDFOR
    \STATE $\mathbf{o}^{p-1}_{t}=\frac{1}{\eta_t}\mathbf{\Gamma }^{-1}_{p-1}(t)\mathbf{b}_{p-1}(t)$ as defined in (\ref{PC-parameter})
    \STATE $\mathbf{\bar{x}}_{t}=-\sqrt{1-\alpha_t}(e^{\eta _t}-1)\mathbf{r}_\mathbf{\theta}(\mathbf{x}_{t-1},\mathbf{z}_{t-1},\mathbf{y},\mathbf{h_x},t-1)$\\
    $+\sqrt{\frac{\alpha_t}{\alpha_{t-1}}}\mathbf{x}_{t-1}-\sqrt{1-\alpha_t}\eta _t \sum_{m=1}^{p-1}\frac{o^{p-1}_{m,t}}{w_{m,t}}\mathbf{D}^x_{m,t}$
    \STATE \textbf{Corrector}:
    \STATE $w_{p,t}=1$, $\mathbf{o}^{p}_{t}=\frac{1}{\eta_t}\mathbf{\Gamma }^{-1}_{p}(t)\mathbf{b}_p(t)$ 
    \STATE ${w_{p,t}}\mathbf{D}^x_{p,t}=\mathbf{s}_\theta(\mathbf{\bar{x}}_{t},,t)-\mathbf{r}_\theta(\mathbf{x}_{t-1},\mathbf{z}_{t-1},\mathbf{y},\mathbf{h_x},t-1)$
    \STATE $\mathbf{{x}}_{t}=-\sqrt{1-\alpha_t}(e^{\eta _t}-1)\mathbf{r}_\mathbf{\theta}(\mathbf{x}_{t-1},\mathbf{z}_{t-1},\mathbf{y},\mathbf{h_x},t-1)$\\
    $+\sqrt{\frac{\alpha_t}{\alpha_{t-1}}}\mathbf{x}_{t-1}-\sqrt{1-\alpha_t}\eta _t \sum_{m=1}^{p}\frac{o^{p}_{m,t}}{w_{m,t}}\mathbf{D}^x_{m,t}$
    
  \end{algorithmic}
  \end{algorithm}

  \begin{algorithm}[t]
    \caption{Interference Cancellation Algorithm of ICDM}
    \label{alg:sampling}
    {\bf {Input:}}
    \small{$\mathbf{y}$, $\mathbf{h_x}$, SINR} \\
    {\bf {Output:}}
    \small{Estimated $\mathbf{\hat{x}}$, $\mathbf{\hat{z}}$}.
    \begin{algorithmic}[1]
      \STATE $\mathbf{x}_0\sim\mathcal{N}(0,\mathbf{I}_{2k})$, $\mathbf{z}_0\sim\mathcal{N}(0,\mathbf{I}_{2k})$
      \FOR {$t=1$ to $p$}
      \STATE $\{\mathbf{r}_\theta(\mathbf{x}_{t-m},\mathbf{z}_{t-m},\mathbf{y},\mathbf{h_x},t-m)\}^t_{m=1}$, \\
      $\{\mathbf{r}_\phi(\mathbf{x}_{t-m},\mathbf{z}_{t-m},\mathbf{y},\mathbf{h_x},t-m)\}^t_{m=1}$ from JCG module
      \STATE $\mathbf{x}_t$ from $t$-order signal PC sampling module
      \STATE $\mathbf{z}_t$ from $t$-order interference PC sampling module
      \ENDFOR
      \FOR {$t=p+1$ to $T$}
      \STATE $\{\mathbf{r}_\theta(\mathbf{x}_{t-m},\mathbf{z}_{t-m},\mathbf{y},\mathbf{h_x},t-m)\}^p_{m=1}$, \\
      $\{\mathbf{r}_\phi(\mathbf{x}_{t-m},\mathbf{z}_{t-m},\mathbf{y},\mathbf{h_x},t-m)\}^p_{m=1}$ from JCG module
      \STATE $\mathbf{x}_t$ from $p$-order signal PC sampling module
      \STATE $\mathbf{z}_t$ from $p$-order interference PC sampling module
      \ENDFOR
      \STATE $\mathbf{\hat{x}}=\mathbf{x}_T$, $\mathbf{\hat{z}}=\mathbf{z}_T$
    \end{algorithmic}
    \end{algorithm}

In addition, the coefficients $o^p_{m,t}$ for $m=1,2,\cdots,p$ need to be recomputed, as they differ from the elements in $\mathbf{o}^{p-1}_t$. Using these data, we can obtain more accurate estimates of $\mathbf{x}_t$ and $\mathbf{z}_t$ with $(p+1)$-th order accuracy as
\begin{align}\label{sampler-corrector}
  &\mathbf{x}_{t}=-\sqrt{1-\alpha_t}(e^{\eta _t}-1)\mathbf{r}_\mathbf{\theta}(\mathbf{x}_{t-1},\mathbf{z}_{t-1},\mathbf{y},\mathbf{h_x},t-1)\nonumber\\
  &+\sqrt{\frac{\alpha_t}{\alpha_{t-1}}}\mathbf{x}_{t-1}-\sqrt{1-\alpha_t}\eta _t \sum_{m=1}^{p}\frac{o^{p}_{m,t}}{w_{m,t}}\mathbf{D}^x_{m,t},\\
  &\mathbf{z}_{t}=-\sqrt{1-\alpha_t}(e^{\eta _t}-1)\mathbf{r}_\mathbf{\theta}(\mathbf{z}_{t-1},\mathbf{z}_{t-1},\mathbf{y},\mathbf{h_x},t-1)\nonumber\\
  &\sqrt{\frac{\alpha_t}{\alpha_{t-1}}}\mathbf{z}_{t-1}-\sqrt{1-\alpha_t}\eta _t \sum_{m=1}^{p}\frac{o^{p}_{m,t}}{w_{m,t}}\mathbf{D}^z_{m,t},
\end{align}
which are the final samples of the ODEs at this time. 

Notably, we use only unconditional gradients for $\mathbf{D}^x_{p,t}$ and $\mathbf{D}^z_{p,t}$ to correct the predicted error of the predicted samples, improving accuracy from $p$-th order to $(p+1)$-th order and thus yielding more realistic samples. As discussed in Remark \ref{hy-par}, a more realistic sample requires samller $\beta$ and $\gamma$. Therefore, we set both to zero, which leads directly to these unconditional gradients. In the initial steps ($t\in[1,p]$), where fewer than $p$ previous samples are available, we employ a $t$-order ConJPC sampler to acquire the current samples.

After the last step, we can obtain the final results $\mathbf{\hat{x}}$ for signal estimation and $\mathbf{\hat{z}}$ for interference estimation, which are also solutions to the MAP problem (\ref{MAP_prb}). The interference cancellation algorithm of ICDM is shown in Algorithm \ref{alg:sampling}. The algorithm is accurate as it jointly solves the ODEs with $(p+1)$-th order accuracy, leveraging precisely estimated gradients from separate training and derivation. It is also fast because the $(p+1)$-th order iteration method reduces the steps needed to acquire the final samples $\mathbf{\hat{x}}$ and $\mathbf{\hat{z}}$. Moreover, the unconditional gradients $\mathbf{D}^x_{p,t}$ and $\mathbf{D}^z_{p,t}$ in the sampling modules can be computed in parallel, further accelerating the algorithm. 

\begin{figure*}[t]
  \centering
  \subfigure[]{\label{CelebAMSE}\includegraphics[width=0.325\textwidth]{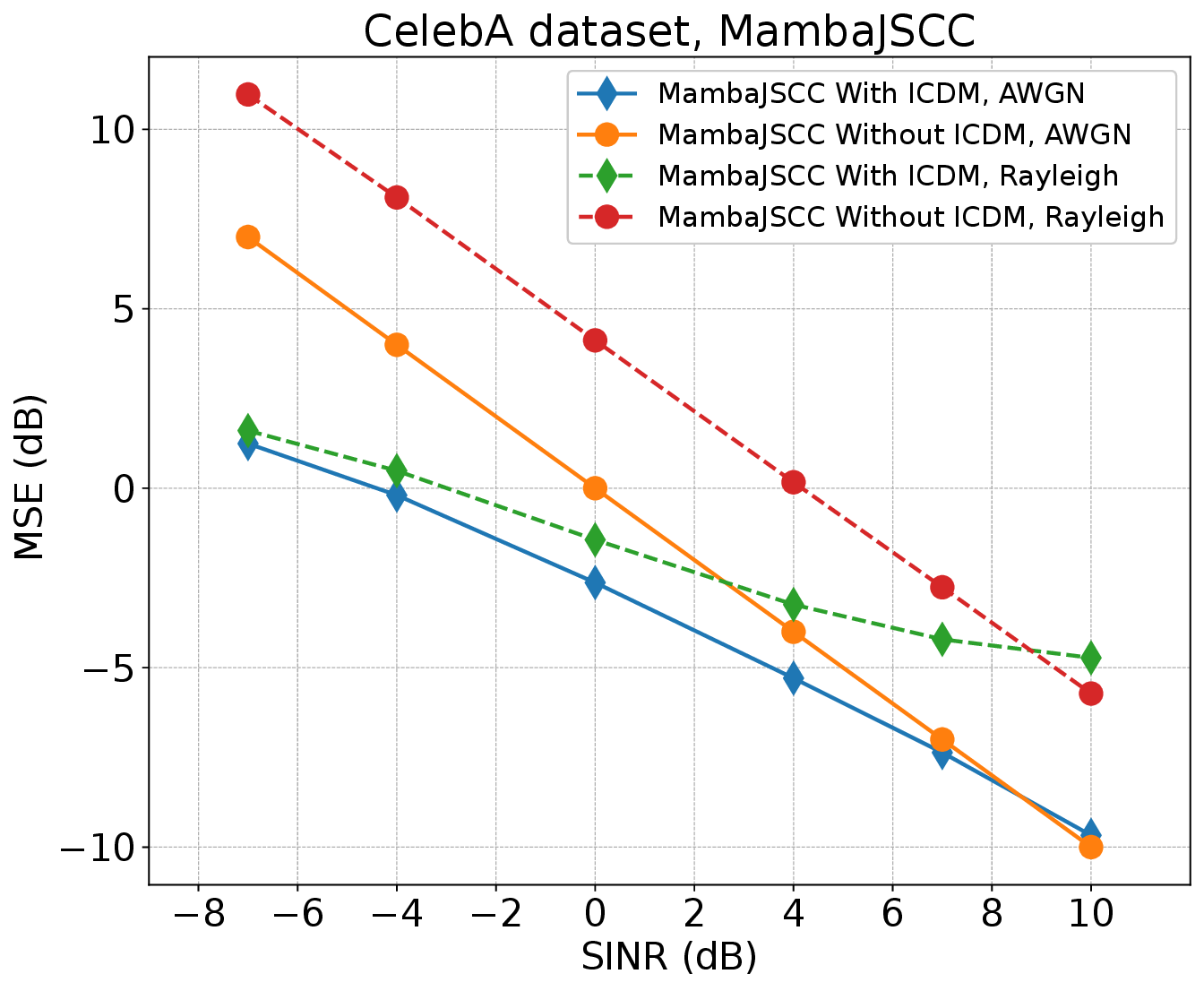}}
  \subfigure[]{\label{CelebALPIPS}\includegraphics[width=0.325\textwidth]{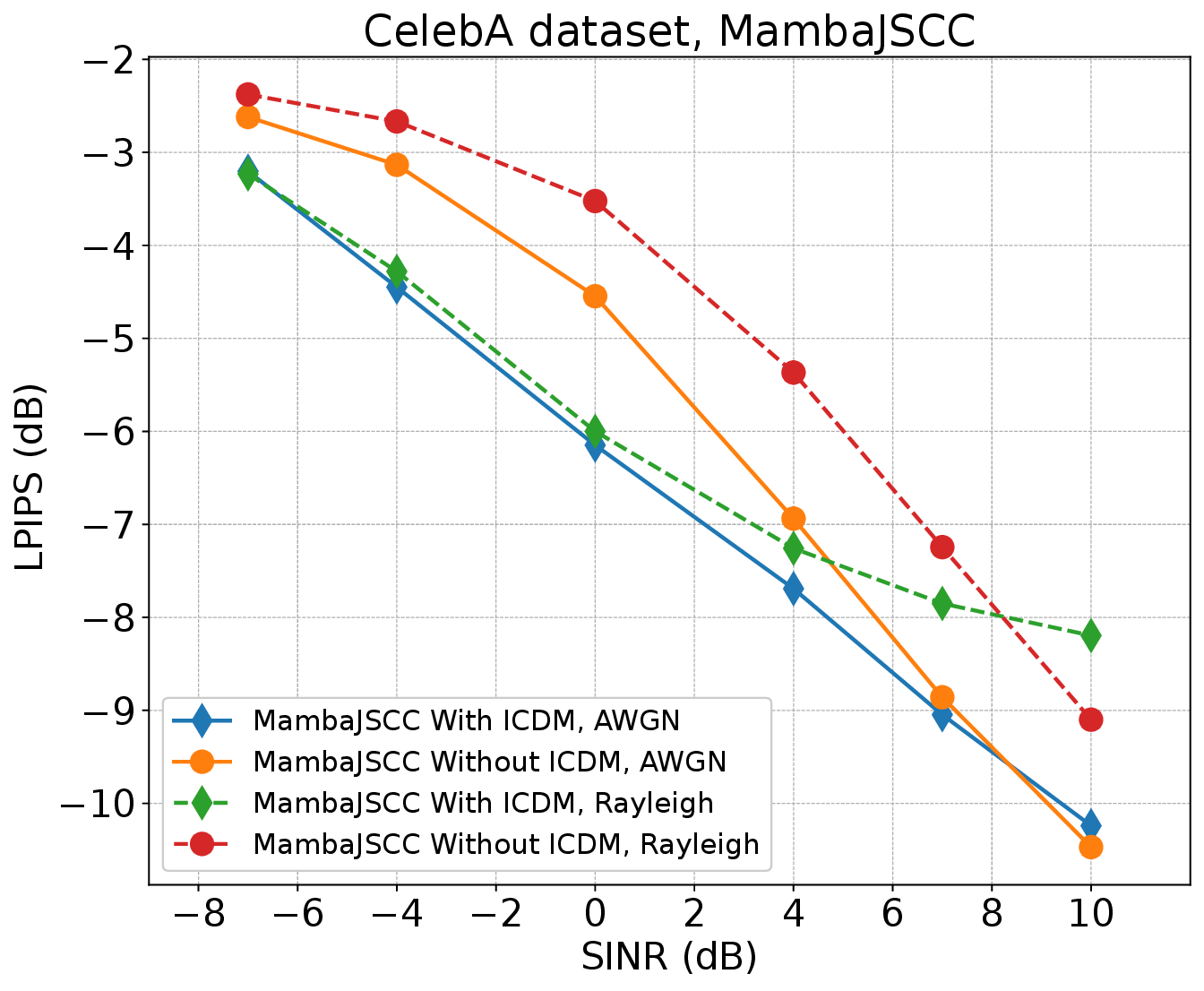}}
  \subfigure[]{\label{CelebACLIP}\includegraphics[width=0.325\textwidth]{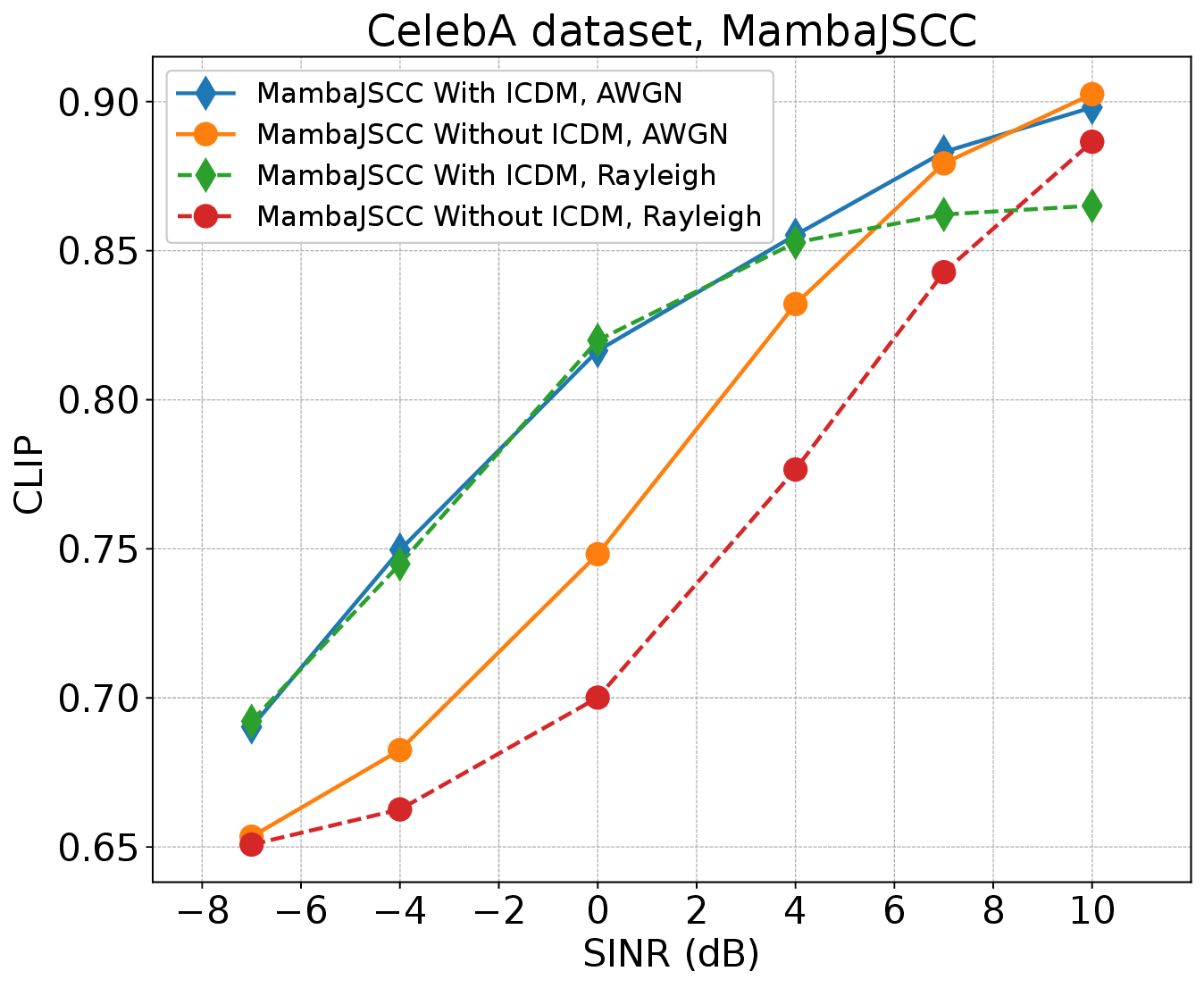}}
  \caption{MSE, LPIPS and CLIP performance of the MambaJSCC-based schemes on the CelebA dataset under both the AWGN and Rayleigh fading channels.}
  \label{CelebA}
  \end{figure*}

  \begin{figure*}[t]
    \centering
    \subfigure[]{\label{BKMSE}\includegraphics[width=0.325\textwidth]{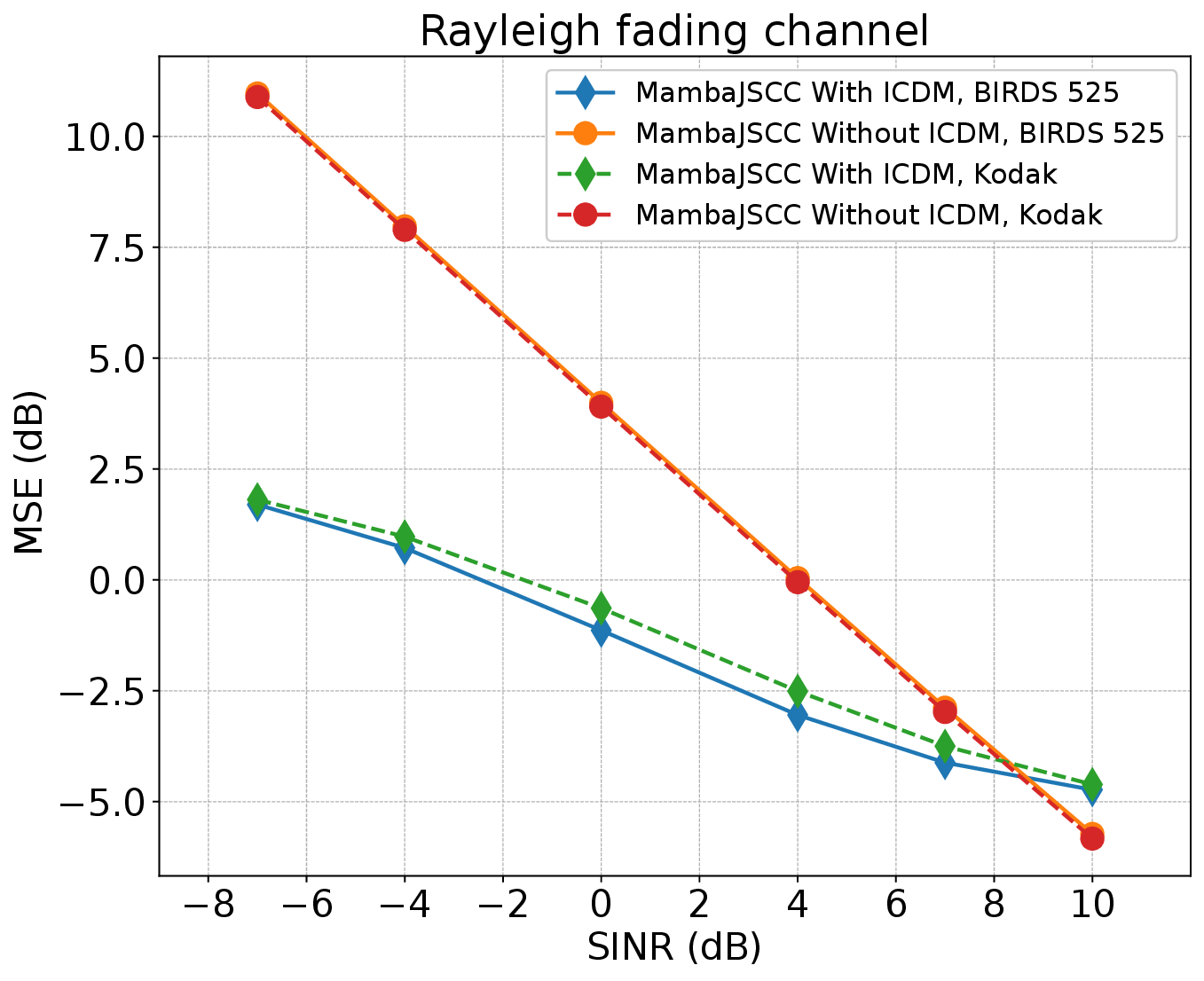}}
    \subfigure[]{\label{BKLPIPS}\includegraphics[width=0.325\textwidth]{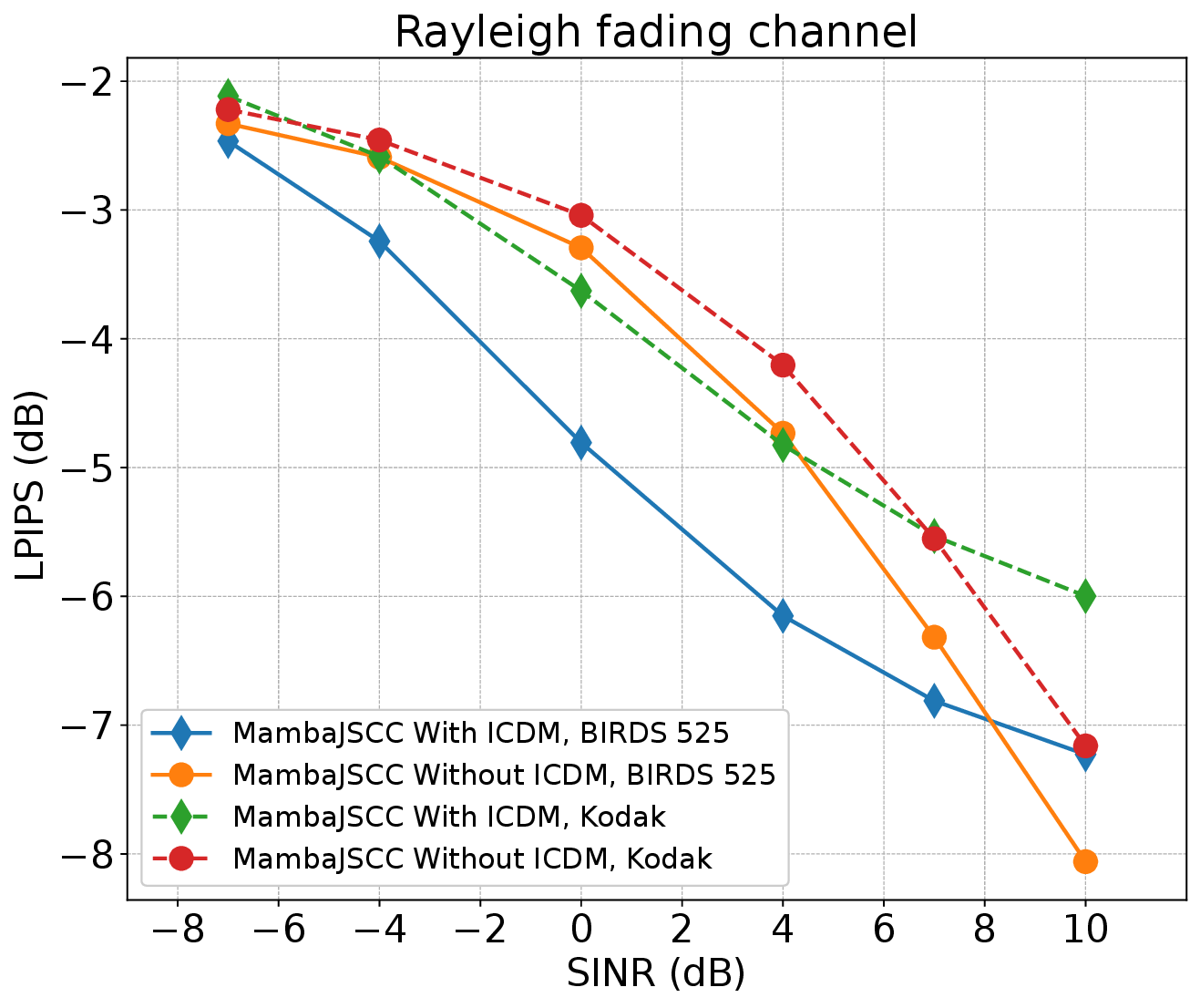}}
    \subfigure[]{\label{BKCLIP}\includegraphics[width=0.325\textwidth]{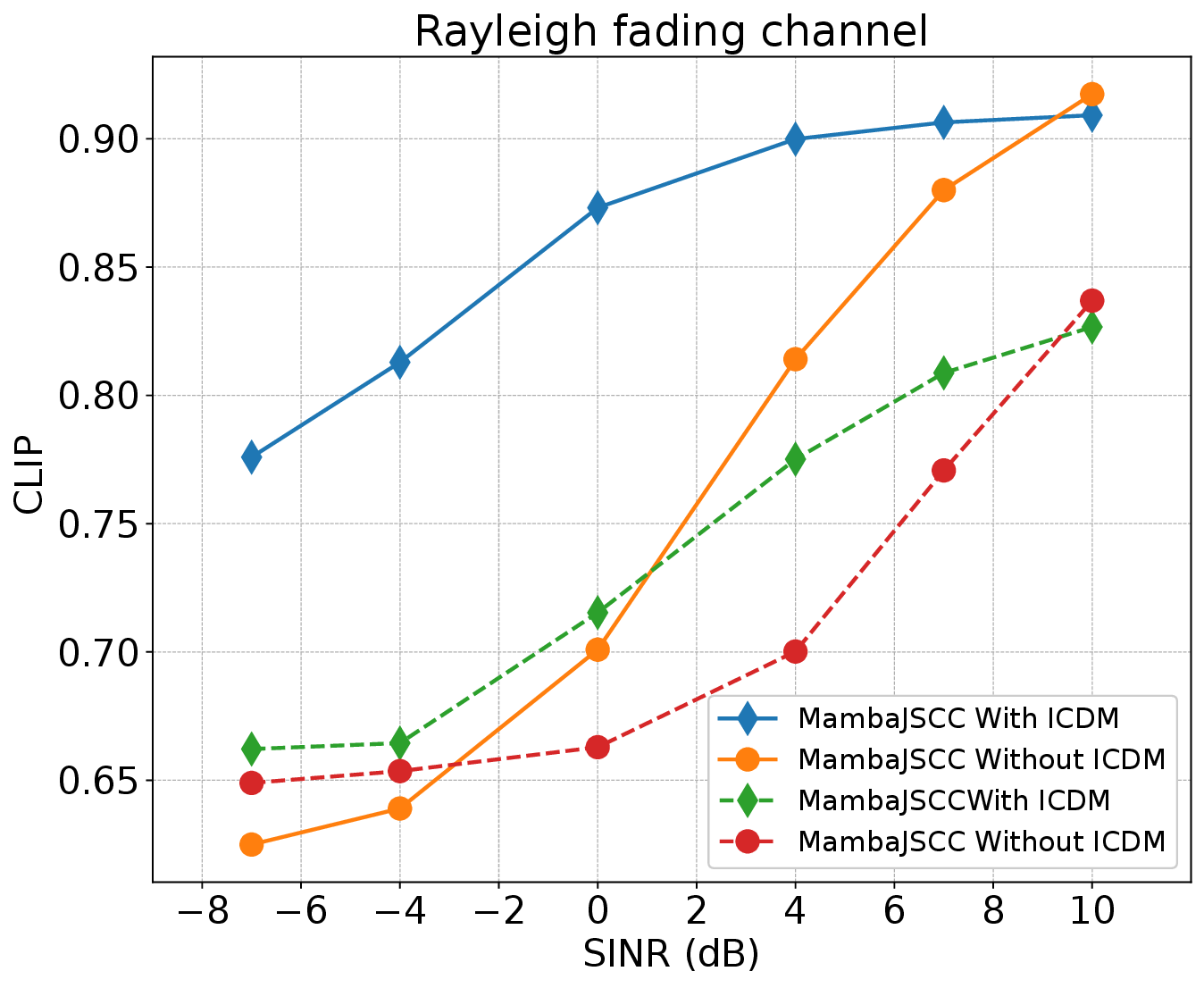}}
    \caption{MSE, LPIPS and CLIP performance of the MambaJSCC-based schemes on the BIRDS 525 and Kodak datasets under the Rayleigh fading channel.}
    \label{BK}
    \end{figure*}
\section{Experimental Results}\label{IV}
In this section, we evaluate the performance of the proposed ICDM based on numerical results.
\subsection{Experimental Setup}
To comprehensively evaluate the proposed ICDM, we employ three datasets: CelebA \cite{CelebA}, BIRDS 525 and Open Image V4 \cite{OpenImg}. The CelebA dataset contains 150,000 face images for training, with 20,000 sampled for evaluation. The BIRDS 525 dataset has 84635 training images of diverse birds in various scenes and 2600 for evaluation. From the Open Image V4 dataset, which contains 9 million images from the real world, we sample 5 million for training and evaluate the performance on the Kodak dataset. During training and evaluation, we resize these images to $128\times 128$ to facilitate implementation without loss of key features.

We evaluate perceptual quality using LPIPS \cite{lpips} and Contrastive Language-Image Pre-training (CLIP) score \cite{clipscore}. The LPIPS is a learned perceptual image patch similarity metric that measures a weighted MSE between deep features of two images extracted by a pretrained VGG16 model \cite{lpips}, while CLIP score computes the cosine similarity between embedding from a pretrained Vision Transformer \cite{clipscore}. For easier comparison in the figure, we convert the mean LPIPS to decibels via $LPIPS(dB)=10\log_{10}({LPIPS})$ \cite{SwinJSCC}. Computational complexity is evaluated via multiply-accumulate operations (MACs), calculated with the Torch Operation Counter library and parameters are counted using the torch library. Inference time is measured on a single NVIDIA RTX 4090 GPU with a batch size of 1. 

We evaluate two advanced JSCC schemes, MambaJSCC \cite{MambaJSCC} and SwinJSCC \cite{SwinJSCC}. We conduct a comparative analysis among the schemes as follows: 1) \textbf{MambaJSCC with ICDM}, where the transmitter uses the MambaJSCC encoder and the decoder first cancels  interference using ICDM and then reconstructs images with the MambaJSCC decoder; 2) \textbf{MambaJSCC without ICDM}, where the decoder skips the interference cancellation step and directly reconstructs images; 3) \textbf{SwinJSCC with ICDM}, where SwinJSCC replaces MambaJSCC in both encoder and decoder while retaining the ICDM step; 4) \textbf{SwinJSCC without ICDM},  defined analogously but without interference cancellation. All JSCC schemes are trained at an signal-to-noise ratio (SNR) of 20 dB without interference, utilizing LPIPS as the loss function to enhance perceptual performance. For each dataset, we train separate models for the AWGN channel and the Rayleigh fading channel to achieve better performance, since the channel type is known during training.

All schemes are evaluated at $\frac{k}{3H_sW_s}=\frac{1}{96}$. The interference is generated by the SwinJSCC scheme with different parameters, deciding the interference pattern. The setting is practical because JSCC schemes are usually pretrained for a specific SNR and channel type to ensure better performance before implementation, and at that time, the interference is unaware. After implementation, even if the decoder detects the presence of interference, it is impractical to fine-tuned the JSCC schemes online. This is because severe interference makes it difficult and resource-consuming to transmit a large number of gradients back to the transmitter. Additionally, to better demonstrate the effectiveness of interference cancellation and avoid confusion with denoising, we consider a scenario where the unware interference mimics Gaussian noise. Therefore, the SNR is set to 20 dB without loss of generality. 

In the proposed ICDM, we train $\mathbf{s}_\mathbf{\theta}(\cdot,t)$ for each JSCC encoder at the transmitter. For a given interference source, we train two $\mathbf{s}_\mathbf{\phi}(\cdot,t)$, one for each channel types, since the channel state should be considered as illustrated in Lemma \ref{lemma1}. We set the number of iteration steps $T=40$ for ICDM. All JSCC schemes and ICDM models are trained with the Adam optimizer with a learning rate of $10^{-4}$ and a batch size of $20$.
\subsection{Performance Evaluation}
\begin{figure*}[t]
  \begin{center}
    \includegraphics[width=0.95\textwidth]{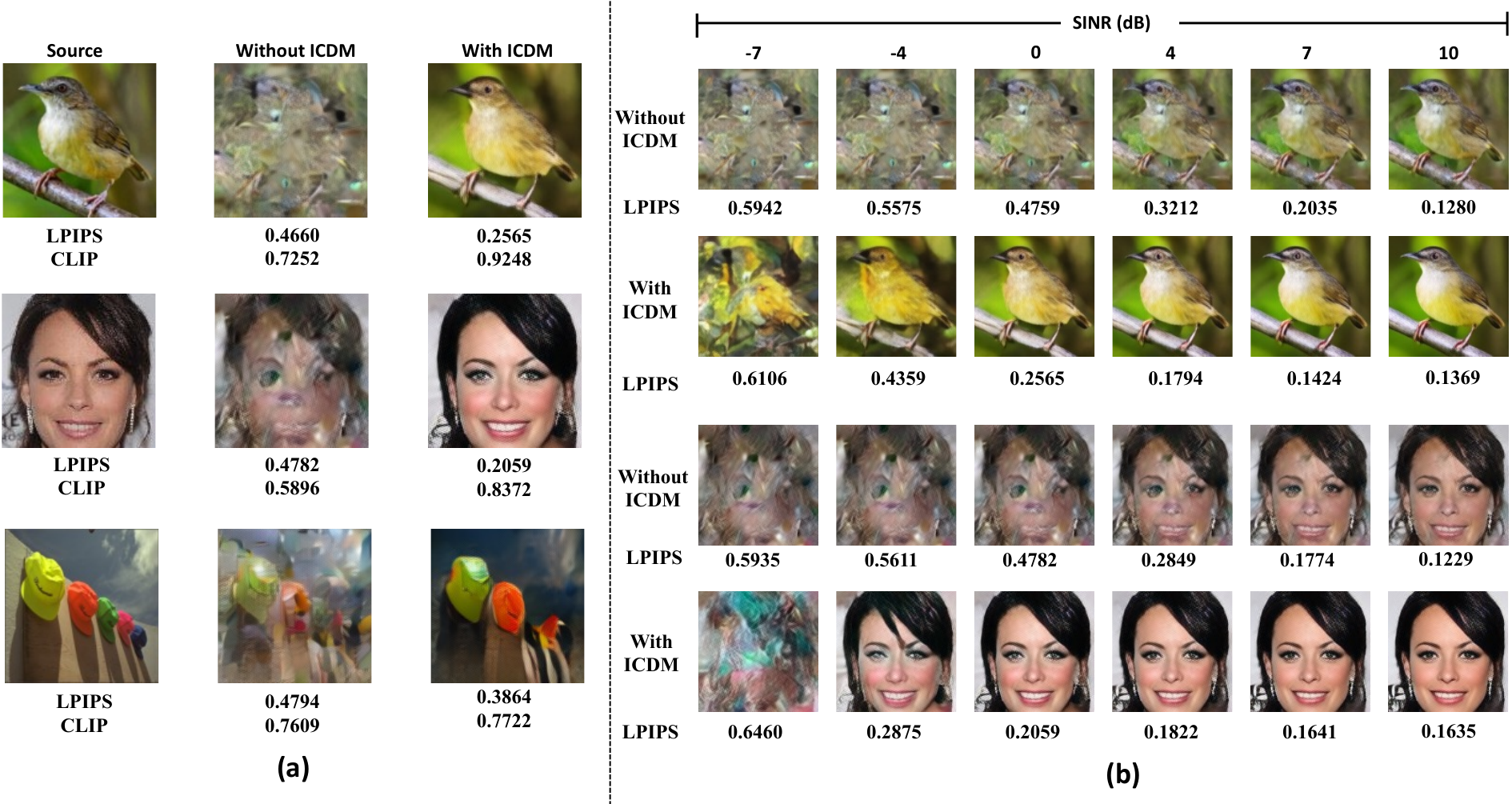}
  \end{center}
    \caption{{(a) Exapmles of the visual comparison between the source images and the images reconstructed with or without the ICDM under the Rayleigh fading channel at SINR=$0$ dB. (b) Exapmles of the visual comparison between the images reconstructed with or without the ICDM versus SINR under the Rayleigh fading channel.}}
    \label{Vis}vvv
\end{figure*}

    \begin{figure*}[t]
      \centering
      \subfigure[]{\label{CelebAMSESwin}\includegraphics[width=0.325\textwidth]{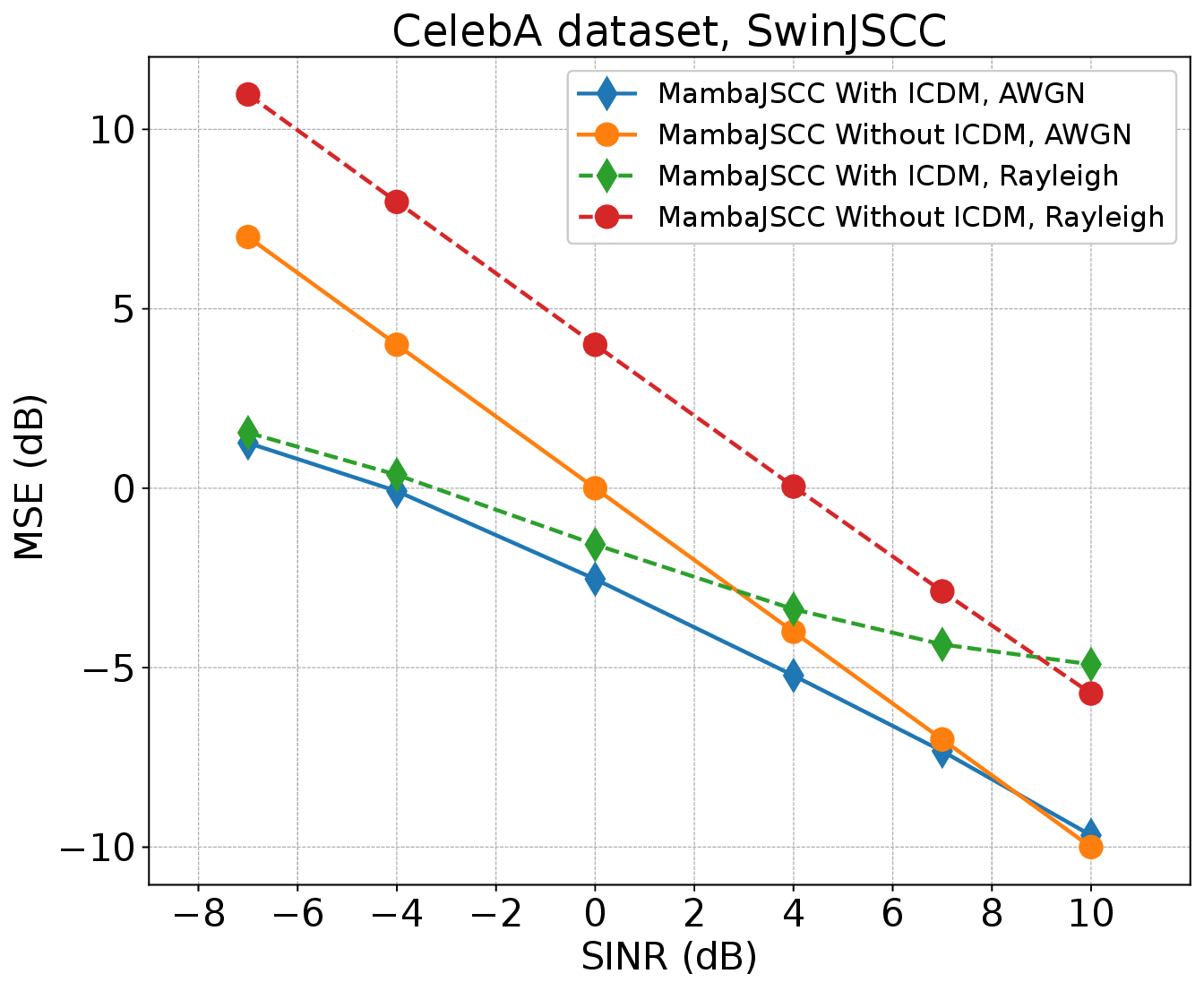}}
      \subfigure[]{\label{CelebALPIPSSwin}\includegraphics[width=0.325\textwidth]{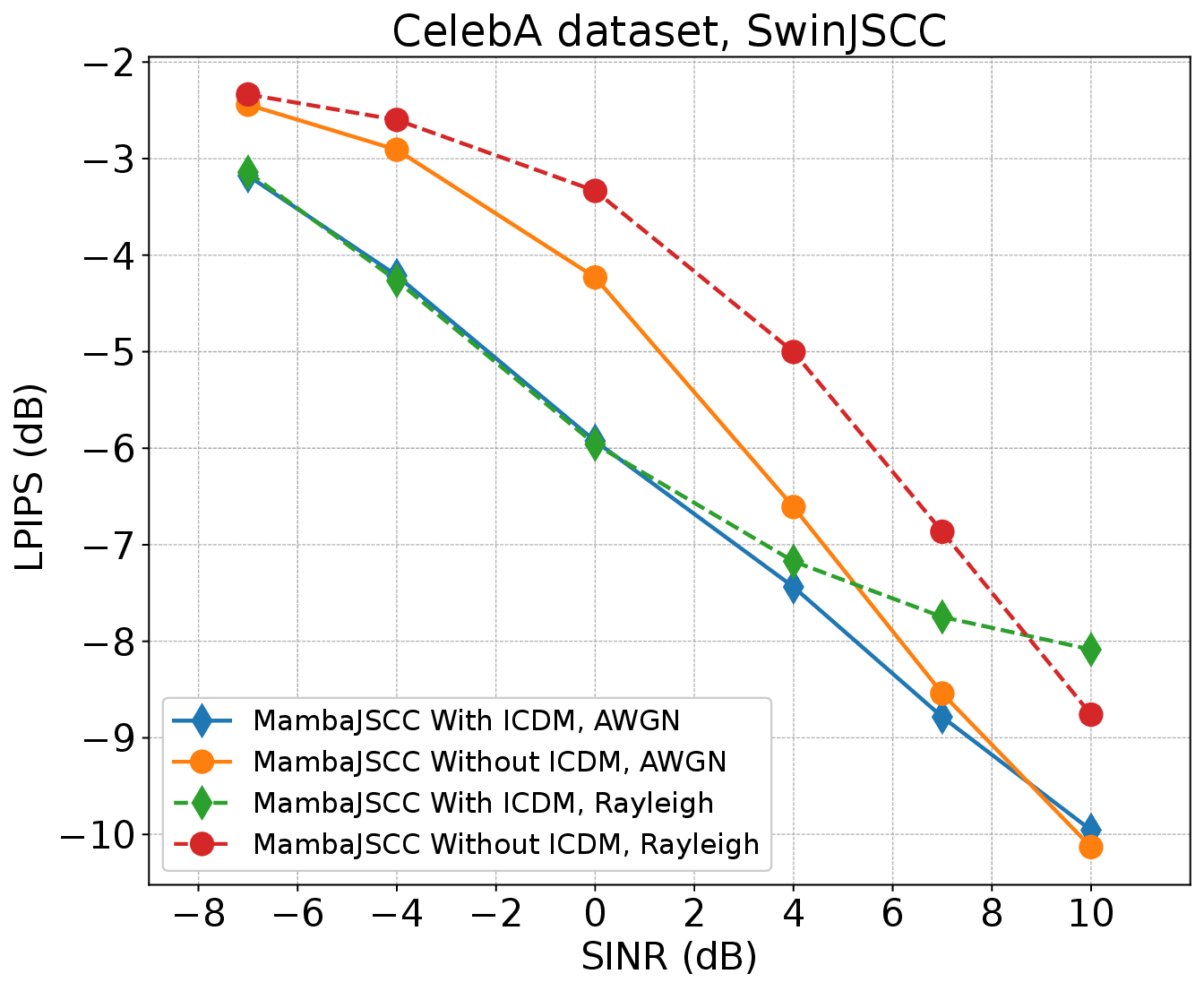}}
      \subfigure[]{\label{CelebACLIPSwin}\includegraphics[width=0.325\textwidth]{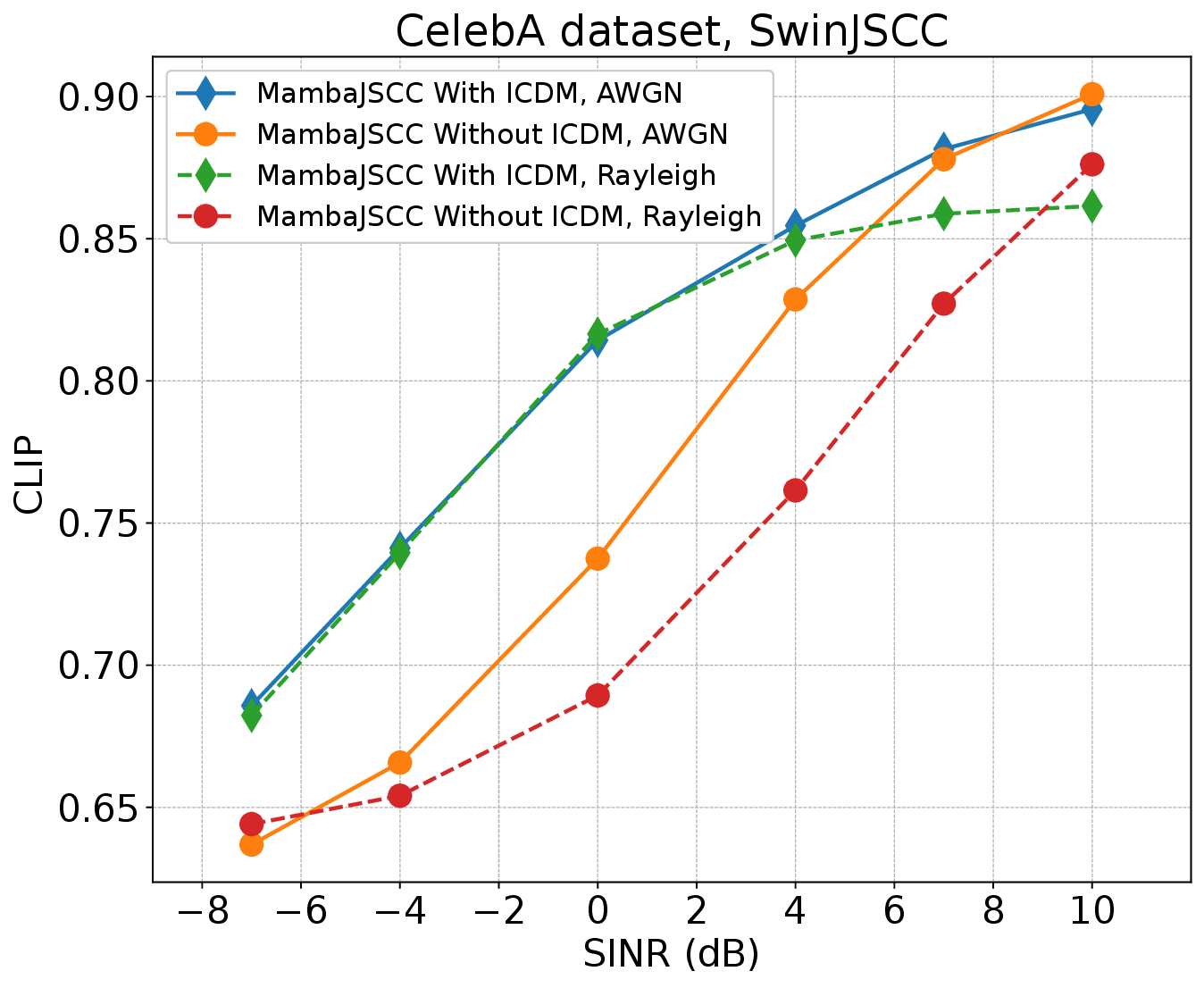}}
      \caption{MSE, LPIPS and CLIP performance of the SwinJSCC-based schemes on the CelebA datasets under the Rayleigh fading channel.}
      \label{CelebASwin}
      \end{figure*}

\begin{figure}[t]
  \begin{center}
    \includegraphics[width=0.45\textwidth]{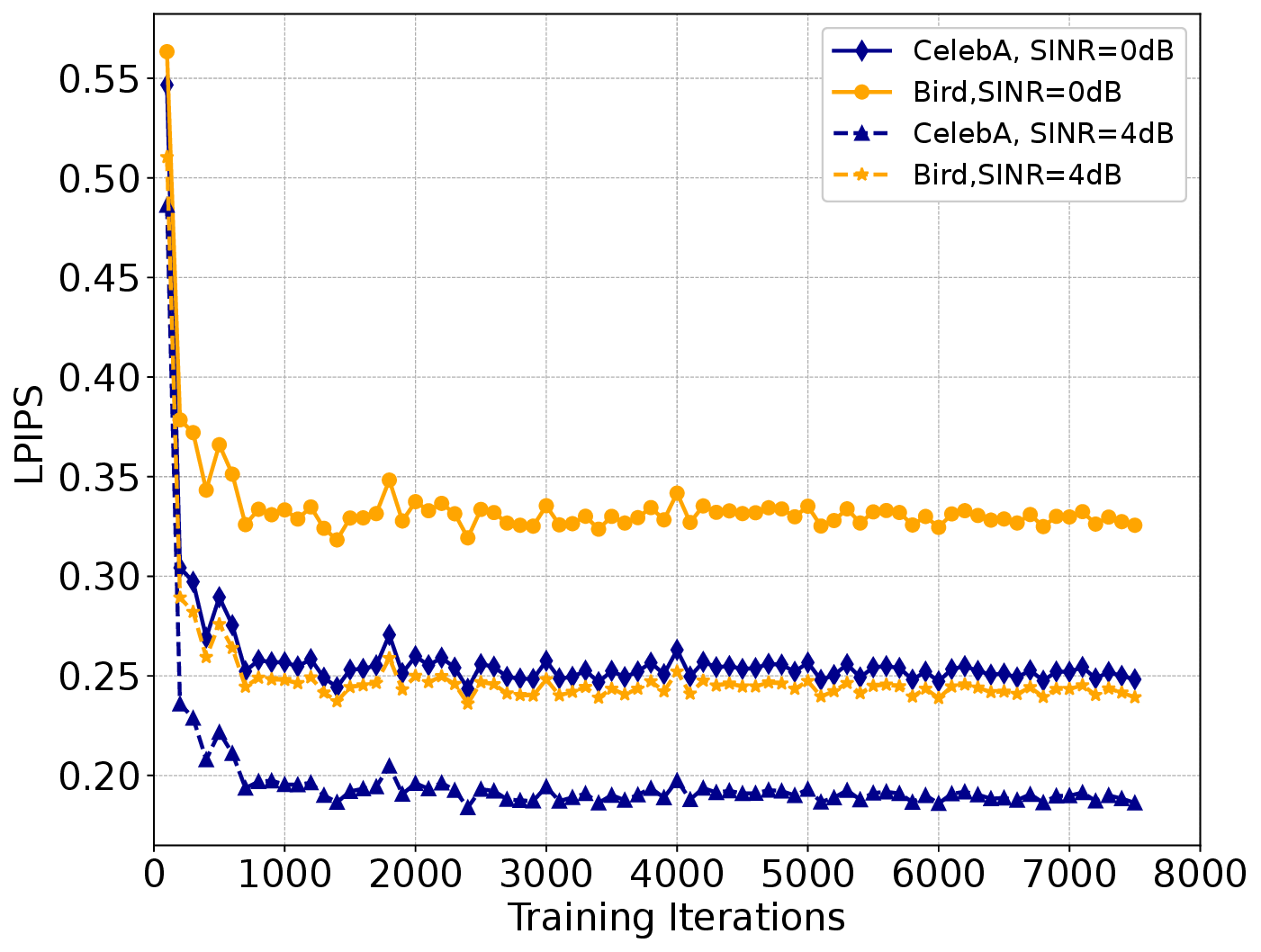}
  \end{center}
    \caption{{LPIPS performance of MambaJSCC with ICDM on the CelebA and BIRDS datasets at SINR=0 and 4 dB under the Rayleigh fading channel.}}
    \label{conv}
\end{figure}

Fig. \ref{CelebA} shows the MSE, LPIPS and CLIP performance of the MambaJSCC-based schemes on the CelebA dataset versus the SNR under both the AWGN and Rayleigh fading channels. It can be observed that the proposed ICDM can significantly improve all the three performance metrics by $4.54$ dB, $2.47$ dB, and $0.12$, respectively, at a signal-to-interference-plus-noise ratio (SINR) of $0$ dB under the Rayleigh fading channel, respectively. Furthermore, the performance gains extend beyond human face images to birds images and even real-world images in the Kodak dataset, as shown in Fig. \ref{BK}. The MambaJSCC-based schemes achieve improvements of $5.11$ dB, $1.51$ dB, and $0.17$ on the BIRDS 525 dataset, and $4.54$ dB, $0.59$ dB, and $0.04$ on the Koda dataset, respectively, at an SINR of $0$ dB under the Rayleigh fading channel.

Fig. \ref{Vis}(a) shows some reconstruction results using MambaJSCC-based schemes for the three datasets under the Rayleigh fading channel at SINR=$0$ dB. It is obvious that the signals are serverly interfered with, making them unrecognizable without ICDM. However, with the proposed ICDM, the reconstructed images are highly realistic and semantically similar to the original ones.

We also conduct experiments on the SwinJSCC-based schemes, as shown in Fig. \ref{CelebASwin}. It can be seen that the SwinJSCC with ICDM scheme significantly outperforms the SwinJSCC without ICDM scheme, achieving improvements of $2.53$ dB, $1.69$ dB, and $0.08$ dB under the AWGN channel, and $5.57$ dB, $2.63$ dB, and $0.13$ dB  under the Rayleigh fading channel at SINR=$0$ dB, respectively. These results demonstrate that the proposed ICDM is suitable for diverse JSCC architectures, regardless of whether the signal and interference originate from the same JSCC framework. This adaptability enables ICDM to support continuously evolving AI foundation model architectures and flexible semantic communication deployment strategies. All the numerical and visual results confirm that the proposed ICDM effectively cancels interference and achieves significant performance gains.

Next, we investigate the characteristics of ICDM revealed by the experiments. For the numerical experimental results, we observe that ICDM achieves its maximum gain at an SINR of 0dB. For example, under a Rayleigh fading channel, this gain reaches $2.47$ dB on the CelebA dataset and $1.57$ dB on the BIRDS 525 dataset. As the SINR deviates from this point, the gain remains significant but gradually diminishes. For example, at SINR = $-4$ dB, the gain reduces to $1.61$ dB on CelebA and $0.65$ dB on BIRDS 525; at SINR = $4$ dB, it decreases to $1.88$ dB on CelebA and $1.42$ dB on BIRDS 525. This behavior stems from ICDM's joint estimation of signal and interference, which effectively leverages learned interference patterns to improve performance. However, it also couples the estimation quality of signal and interference, so that a  degradation in one leads to a decrease in the other. The experiment shown in Fig. \ref{conv} demonstrates this relationship. It polts the performance of ICDM with a well-trained $\mathbf{s}_\theta(\cdot,t)$ under the Rayleigh fading channel of an SINR of $0$ dB and $4$ dB, as a function of the training iterations of $\mathbf{s}_\phi(\cdot,t)$. As $\mathbf{s}_\phi(\cdot,t)$ continuous to learn the interference pattern, i.e., the discrepancy between the learned and true interference patterns decreases, the quality of the estimated interference signal improves, leading to a steady enhancement of the overall system performance. This confirms that accurate interference estimation is crucial for the high quality image reconstruction from the estimated transmitted signal. When SINR is around $0$ dB, the powers of interference and signal are comparable, so both estimations are of both high quality, resulting in the maximum gain. As SINR deviates from $0$ dB, an improvement in the estimation of one component is offset by a degradation in the other, reducing the overall performance gain. Consequently, ICDM achieves its peak gain at SINR = $0$ dB, and the gain diminishes as SINR increases or decreases.

Furthermore, when SINR decreases to $-7$ dB or increases to $10$ dB, the proposed ICDM cannot effectively eliminate interference. The fact that ICDM is only effective within SINR range of $[-4, 7]$ dB is therefore understandable and acceptable. This is more intuitively illustrated in Fig. \ref{Vis}(b). At SINR = $-7$ dB, although ICDM slightly improves perceptual quality, the reconstructed images remain unidentifiable and the system is still unusable, because the interference is simply too strong to extract a meaningful signal. At SINR = $10$ dB, the images reconstructed without ICDM are already clear and of high quality, indicating that interference has little impact and that ICDM's interference removal is unnecessary. However, within the SINR range of $[-4,7]$, ICDM's interference cancellation is highly effective. Without ICDM, the reconstructed images are heavily distorted by interference and unrecognizable. In contrast, with ICDM the interference is dramatically reduced, yielding reconstructed images that are realistic, visually appealing, and semantically faithful to the originals.

Fig. \ref{conv} also demonstrates that $\mathbf{s}_\phi(\cdot,t)$ in ICDM requires only abou 600 training iterations to achieve substantial performance gain, showing that the proposed ICDM can quickly learn interference patterns and adapts to varying conditions. Moreover, ICDM significantly reduces the inference time for the sampling process. As shown in Fig. \ref{IT}, ICDM needs only about $40$ sampling steps and $1.57$s to estimate the signal and interference, while the iterative sampling methods introduced in Diffusion Transformer (DiT) \cite{dit} and denoising diffusion probabilistic model (DDPM) \cite{ddpm} use $200$ and $1000$ iteration steps, requiring $7.8$s and $39.2$s for inference, respectively. This acceleration stems from ICDM's use of a $p$-order ConJPC sampler, which solves the ODEs with $(p+1)$-th order accuracy, compared with the simple Euler methods used by the other approaches. Consequently, ICDM converges to a solution much more rapidly.
\begin{figure}[t]
  \begin{center}
    \includegraphics[width=0.408\textwidth]{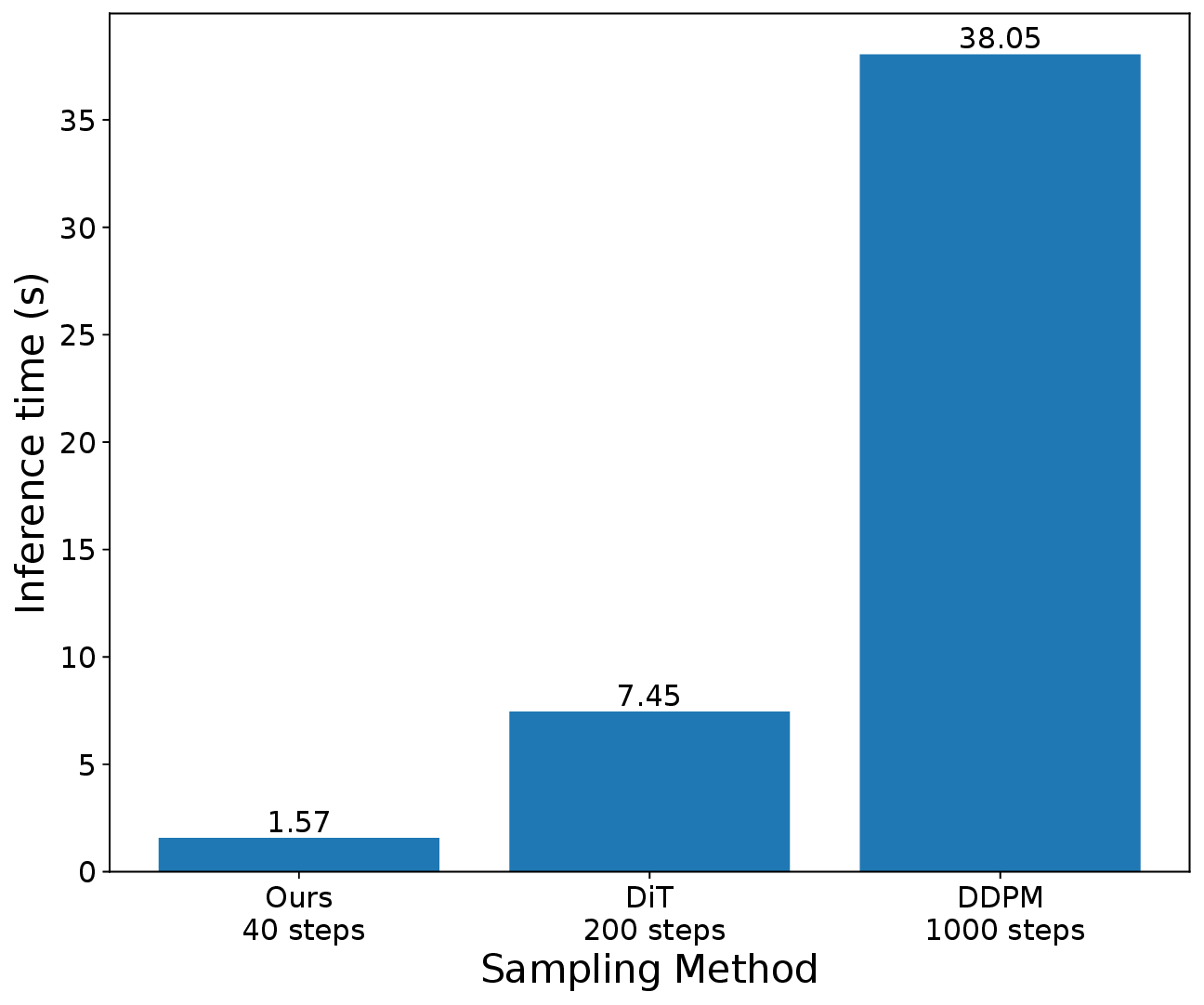}
  \end{center}
    \caption{{Inference time of different sampling methods as recommended in the papers.}}
    \label{IT}
\end{figure}

\subsection{Ablation Study}
For solving the MAP problem formulated in (\ref{MAP_prb}), estimating the gradients of the channel transition probability $\nabla_{\mathbf{x}_t} p_{\mathbf{y}|\mathbf{x}_t,\mathbf{z}_t,\mathbf{h_x}}(\mathbf{y}|\mathbf{x}_t,\mathbf{z}_t,\mathbf{h_x})$ is crucial. Therefore, we compare our estimation with some previous proposed methods, including DPS \cite{dps}, GDM \cite{gdm} and the Projection method \cite{projection}. In both DPS and GDM, $\mathbf{\hat{x}}_0$ and $\mathbf{\hat{z}}_0$ are first estimated from the samples $\mathbf{{x}}_{t-1}$ and $\mathbf{{z}}_{t-1}$ as 
\begin{align}
  \mathbf{\hat{x}}_T&=\frac{1}{\sqrt{\alpha_{t-1}}}\left(\mathbf{x}_{t-1}-\sqrt{1-\alpha_{t-1}}\mathbf{s}_\theta(\mathbf{x}_{t-1},t-1)\right),\nonumber\\
  \mathbf{\hat{z}}_T&=\frac{1}{\sqrt{\alpha_{t-1}}}\left(\mathbf{z}_{t-1}-\sqrt{1-\alpha_{t-1}}\mathbf{s}_\phi(\mathbf{z}_{t-1},t-1)\right),
\end{align}
and then DPS estimate the gradients as 
\begin{align}\label{DPS}
\mathbf{\bar{r}}_x(\cdot)=\kappa_t\nabla_{\mathbf{\hat{x}}_{T}} ||\mathbf{y}-\sqrt{P_x}\mathbf{W_x\hat{x}}_T-\sqrt{P_z}\mathbf{W_z\hat{z}}_T||^2_2,\\
\mathbf{\bar{r}}_z(\cdot)=\kappa_t\nabla_{\mathbf{\hat{z}}_{T}} ||\mathbf{y}-\sqrt{P_x}\mathbf{W_x\hat{x}}_T-\sqrt{P_z}\mathbf{W_z\hat{z}}_T||^2_2,\nonumber
\end{align}
where $\kappa_t=\frac{1}{2||\mathbf{y}-\sqrt{P_x}\mathbf{W_x\hat{x}}_T-\sqrt{P_z}\mathbf{W_z\hat{z}}_T||}$.

The GDM method is similar to the DPS method, but it uses the gradient of the transition probability as
\begin{align}\label{GDM}
  &\mathbf{\bar{r}}_x(\cdot)=\sqrt{P_x}\mathbf{W_x}(\mathbf{y}-\sqrt{P_x}\mathbf{W_x\hat{x}}_T-\sqrt{P_z}\mathbf{W_z\hat{z}}_T)\\
  &(\mathbf{W_x}\mathbf{W_x}^T+\mathbf{W_z}\mathbf{W_z}^T+\frac{(2-\alpha_{t-1})||\mathbf{y}||^2_2}{1-\alpha_{t-1}}\mathbf{I})^{-1},\nonumber\\
  &\mathbf{\bar{r}}_z(\cdot)=\sqrt{P_z}\mathbf{W_z}(\mathbf{y}-\sqrt{P_x}\mathbf{W_x\hat{x}}_T-\sqrt{P_z}\mathbf{W_z\hat{z}}_T)\nonumber\\
  &(\mathbf{W_x}\mathbf{W_x}^T+\mathbf{W_z}\mathbf{W_z}^T+\frac{(2-\alpha_{t-1})||\mathbf{y}||^2_2}{1-\alpha_{t-1}}\mathbf{I})^{-1}.\nonumber
\end{align}

In the Projection method, the projection adds standard Gaussian noise $\mathbf{\epsilon}_y$ to the received signal $\mathbf{y}$ as $\mathbf{y}_{t-1}$ by
\begin{align}
  \mathbf{y}_{t-1}=\sqrt{\alpha_{t-1}}\mathbf{y}+\sqrt{1-\alpha_{t-1}}\mathbf{\epsilon}_y,
\end{align}
and then compute the gradients as
\begin{align}\label{Projection}
\mathbf{\bar{r}}_x(\cdot)=\nabla_{\mathbf{x}_{t-1}} ||\mathbf{y}_{t-1}-\sqrt{P_x}\mathbf{W_x{x}}_{t-1}-\sqrt{P_z}\mathbf{W_z{z}}_{t-1}||^2_2,\\
\mathbf{\bar{r}}_z(\cdot)=\nabla_{\mathbf{z}_{t-1}} ||\mathbf{y}_{t-1}-\sqrt{P_x}\mathbf{W_x{x}}_{t-1}-\sqrt{P_z}\mathbf{W_z{z}}_{t-1}||^2_2.\nonumber
\end{align}

We apply these estimation methods to ICDM and evaluate the performance on the MambaJSCC with ICDM scheme. The results is shown in Table. \ref{eg}. The baseline represents the performance of MambaJSCC without ICDM scheme. The proposed method achieves the best performance because it estimates the gradient using the current samples and received signal under accurate assumptions about the transmitted signal and interference and theoretical derivation. In contrast, the DPS and GDM methods perform considerably worse: they compute $\mathbf{\hat{x}}_{T}$ and $\mathbf{\hat{z}}_{T}$ from $\mathbf{x}_{t-1}$ and $\mathbf{z}_{t-1}$ in a single estimation step, introducing larger error and degrading performance. The Projection method performs the worst of all, since adding noise to the received signal causes loss of information.

Next, we compare the performance of our ICDM and the CDDM proposed in \cite{CDDM} with similar denoising steps under the Rayleigh fading channel with SINR=0 dB based on MambaJSCC. The baseline represents the performance of MambaJSCC without ICDM or CDDM. The CDDM here learns the distribution of the transmitted signal and selects denoising steps according to the SINR. Its total diffusion steps are set to 200 because here the CDDM is also built upon DiT for fairness and the recommended diffusion steps in DiT model are 200.
It could be seen from Table \ref{tab1} that CDDM can enhance the perceptual quality but still performs worse than ICDM in both datasets even with more iteration steps. This is because the CDDM directly treats the interference as Gaussian noise, which is not accurate. While our ICDM learns the interference pattern and achieve interference cancellation with a theoretically designed method.
\begin{table}[tbp]
  \centering
  \caption{MSE, LPIPS and CLIP performance and required sampling iteration steps for interference cancellation using different guided-gradient estimation methods for the transition probability funtion.}
  \label{eg}
  \resizebox{\linewidth}{!}
  {
  \begin{tabular}{|c|c|c|c|c|}
    \hline
    Method & MSE (dB)& LPIPS & CLIP & Iteration Steps\\\hline
    Baseline & 0.72 & 0.44 & 0.70 & -\\
    DPS & 1.30 & 0.44 & 0.72 & 60\\
    GDM & 1.22 & 0.39 & 0.74 & 40\\
    Projection & 1.46 & 0.53 & 0.67 & 60\\
    Ours & \textbf{0.72} & \textbf{0.25} & \textbf{0.82} & 40\\\hline
  \end{tabular}
  }
\end{table}

\begin{table}[t]
  \renewcommand\arraystretch{1.5}
  \centering
  \caption{{The LPIPS, CLIP performance and Iteration steps of CDDM and ICDM under the Rayleigh fading channel with SINR=0 dB.}}
  \label{tab1}
  \resizebox{\linewidth}{!}{
  \begin{tabular}{|c|c|c|c|c|c|}\hline
  \multirow[vpos]{2}{*}{Method} & \multicolumn{2}{c|}{CelebA} & \multicolumn{2}{c|}{BIRDS}&\multirow[vpos]{2}{*}{Iteration Steps}\\ \cline{2-5}
  & LPIPS & CLIP &LPIPS & CLIP&\\
  \hline
  Baseline& 0.55 & 0.70 & 0.47 & 0.70 & -\\\hline
  CDDM& 0.37 & 0.73 & 0.38 & 0.80 & 116\\\hline
  ICDM& 0.25 & 0.82 & 0.33 & 0.87 & 40\\\hline

  \end{tabular}}
  \end{table}

\section{Conclusion}\label{V}
In this paper, we propose a novel interference cancellation paradigm by formulating the problem as a MAP estimation over the joint posterior distribution of signal and interference. We theoretically prove that its solution yields  high-quality estimates of both components. To solve this problem effectively, we develop ICDM, which accurately estimates the log-gradient of the joint posterior using diffusion models and Bayes' theorem, and integrates it with a $p$-order ConJPC sampler. The proposed ICDM is well-suited to practical scenarios. Extensive experiments confirm that ICDM significantly reduces MSE and improves the perceptual quality of reconstructed images while requiring less training and inference time. For example, on the CelebA dataset under the Rayleigh fading channel with SNR=$20$ dB and SINR=$0$ dB, ICDM reduces MSE by $5.67$ dB, improves LPIPS by $2.65$ dB, and increases the CLIP score from $0.7$ to $0.82$ compared with schemes without ICDM. 
\bibliography{ref}

\end{document}